\pdfoutput=1
\documentclass[affil-it,auth-lg,british,compression, contents, custom-packages={},references={bigone}]{lib/preprint}

\begin{document}
	
	% Hypersetup
	\hypersetup{
		pdftitle = {Adjusting Higher Chern--Simons Theory},
		pdfauthor = {Gianni Gagliardo,Dominik Rist,Christian Saemann,Martin Wolf},
		pdfkeywords = {}
	}
	
	% Date of preprint
	\date{\today}
	
	% All emails in the order of appearance
	\email{ggg2001@hw.ac.uk, dr.dominik.rist@gmail.com,c.saemann@hw.ac.uk,m.wolf@ surrey.ac.uk}
	
	% All preprint numbers in the order of appearance
	\preprint{DMUS--MP--25/02}
	
	% Title
	\title{Adjusting Higher Chern--Simons Theory}
	
	% All authors
	\author[a]{Gianni~Gagliardo\,\orcidlink{0009-0005-5724-7965}\,}
	\author[a]{Dominik~Rist\,\orcidlink{0000-0002-1817-3458}\,}
	\author[a]{\\ Christian~Saemann\,\orcidlink{0000-0002-5273-3359}\,}
	\author[b]{Martin~Wolf\,\orcidlink{0009-0002-8192-3124}\,}
	
	% All affiliations
	\affil[a]{Maxwell Institute for Mathematical Sciences\\ Department of Mathematics, Heriot--Watt University\\ Edinburgh EH14 4AS, United Kingdom}
	\affil[b]{School of Mathematics and Physics,\\ University of Surrey, Guildford GU2 7XH, United Kingdom}
	
	% Abstract
	\abstract{A fundamental problem in formulating higher Chern--Simons theories is the construction of a consistent higher gauge theory that circumvents the fake-flatness constraint. Here, we propose a solution to this problem using adjusted higher connections. Concretely, we shall demonstrate that there is an obstruction to the direct construction of such action functionals since, generically, adjusted higher gauge algebras do not admit an inner product. To overcome this obstruction, we introduce half-adjusted higher Chern--Simons theories. These theories have both well-defined underlying kinematic data as well as the expected properties of a higher generalisation of Chern--Simons theory. We develop the general construction of these theories in arbitrary dimensions and provide explicit details for the four-dimensional case. We also present the complete differential cohomological framework for principal 2-bundles with half-adjusted connections. Finally, we discuss an alternative approach introducing additional trivial symmetries.}
	
	% Acknowledgements
	\acknowledgements{We gratefully acknowledge discussions with Alexander Schenkel and Beno{\^i}t Vicedo.}
	
	% Declarations
	\declarations{
		\textbf{Funding.} 
		GG has been supported by an STFC studentship ST/Y509206/1.\\[5pt]
		\textbf{Conflict of interest.}
		The authors have no relevant financial or non-financial interests to disclose.\\[5pt]
		\textbf{Data statement.}
		No additional research data beyond the data presented and cited in this work are needed to validate the research findings in this work.\\[5pt]
		\textbf{Licence statement.}
		For the purpose of open access, the authors have applied a Creative Commons Attribution (CC-BY) license to any author-accepted manuscript version arising.
	}
	
	% Body
	\begin{body}
		
		\section{Introduction and conclusions}
		
		Besides Yang--Mills theory, Chern--Simons theory~\cite{Chern:1974ft} is arguably the second most extensively studied gauge theory. Mathematically, it has achieved remarkable success as a computational tool for topological problems, particularly in computing knot and link invariants~\cite{Witten:1988hf}. Physically, it has emerged in diverse contexts ranging from M2-brane models~\cite{Bagger:2007jr,Gustavsson:2007vu} to colour--kinematics duality~\cite{Ben-Shahar:2021zww,Borsten:2022vtg}. In more general forms, holomorphic Chern--Simons theory has played a vital role in topological and twistor string theory~\cite{Witten:9207094,Witten:2003nn}, whilst Chern--Simons-like actions underlie the pure spinor superfield formulation~\cite{Cederwall:2013vba}. More recently, Costello constructed a four-dimensional Chern--Simons theory in~\cite{Costello:2013zra}, which was then further developed in~\cite{Costello:2017dso,Costello:2018gyb,Costello:2019tri}. This theory has attracted considerable attention since it provides an elegant gauge-theoretic perspective on the Yang--Baxter equations and two-dimensional integrable systems.
		
		String-theory-inspired field theories now incorporate connection forms of higher degrees. These combine into the kinematic data of higher gauge theories, featuring higher or categorified gauge groups and describing higher-dimensional parallel transport. The natural occurrence of higher gauge theories raises the question of whether higher forms of Chern--Simons theory can replicate the successes of ordinary Chern--Simons theories for higher-dimensional objects.
		
		From another perspective, Chern--Simons-type actions are particularly appealing examples of higher gauge theories to study because their action functionals are mostly fixed by categorification. Contrary to this, higher analogues of Yang--Mills theory are much less canonical and involve many choices. Moreover, the seemingly physically preferred analogues of Yang--Mills theory, i.e.~the ones appearing in the context of string theory and supergravity, seem to rely on an inner product structure of the higher gauge Lie algebra which deviates from the mathematically preferred choice.\footnote{See the discussion in \ref{ssec:general_higher_gauge}.}
		
		Early versions of higher Chern--Simons theory were implicit in the AKSZ approach~\cite{Alexandrov:1995kv}. These were then followed by studies of BF-type higher gauge theories~\cite{Girelli:2003ev,Girelli:0708.3051,Martins:2010ry}, i.e.~higher Chern--Simons theories with strict higher gauge group, which can be quantised using generalised spin foam quantisation. A main motivation for studying these higher gauge theories stems from their use in a formulation of a theory of quantum gravity, see e.g.~\cite{Radenkovic:2019qme} and~\cite{Stipsic:2025rhk} for very recent work.
		
		A general but abstract description of higher Chern--Simons theory appeared in~\cite{Fiorenza:2011jr}, with related works including~\cite{Antoniadis:2013jja} and subsequent developments~\cite{Song:2022ftp,Song:2023rra}. For developments motivated by the applications of higher Chern--Simons theories to categorified knot invariants, see e.g.~the works~\cite{Soncini:2014ara,Zucchini:2015ohw,Zucchini:2019mbz,Zucchini:2021bnn}. A combinatorial approach to quantising higher Chern--Simons theory was discussed in~\cite{Chen:2025msk,Chen:2025qpt}.
		
		Recent work~\cite{Schenkel:2024dcd,Chen:2024axr} has considered higher analogues of Costello's four-dimensional Chern--Simons theory. As explained in~\cite{Schenkel:2024dcd}, the standard description of four-dimensional integrable systems in terms of a six-dimensional holomorphic Chern--Simons perspective is somewhat unsatisfying, since one would expect Lax connections for a $(d+1)$-dimensional integrable field theory to involve a $d$-form. This suggests that even three-dimensional integrable field theories should involve higher gauge connections.
		
		In this paper, we seek to define mathematically and physically interesting higher Chern--Simons theories. Specifically, we have the following expectations of such theories:
		\begin{enumerate}[(i)]
			\itemsep-1pt
			\item The theory should be a higher gauge theory with kinematic data consisting exclusively of a connection on a higher principal bundle with consistently acting gauge and higher gauge transformations. This ensures that the theory can be consistently formulated on any manifold of suitable dimension.
			\item The equations of motion of the theory's action functional should be equivalent to all curvature forms vanishing identically. 
			\item The differential of the Lagrangian form, when lifted to higher dimensions, should be a polynomial expression in the curvatures and hence a topological invariant.
		\end{enumerate}
		Note that as common in the discussion of Chern--Simons theory, we will restrict ourselves here to higher principal bundles which are topologically trivial.
		
		A fundamental challenge in higher gauge theories\footnote{which is less visibly shared with gauge-matter theories} is the subtle definition of the curvature forms of connections on higher principal bundles. These forms are vital, as they also induce the form of gauge transformations.\footnote{Gauge transformations connected to the identity are partially flat homotopies, cf.~\ref{ssec:four_dimensional_case_infinitesimal_considerations}.} For ordinary connections, the curvature form is directly determined by the Maurer--Cartan equation
		\begin{equation}
			F\ =\ \rmd A+\tfrac12[A,A]\ =\ 0
		\end{equation}
		for a Lie-algebra valued gauge potential one-form $A$. In higher gauge theories, this uniqueness is lost: any higher curvature form can be modified by expressions involving lower curvature forms without affecting the moduli space of flat connections. For instance, the redefinition of a 3-form curvature $H$,
		\begin{equation}
			H\ \rightarrow\ H+\kappa(A,F)
		\end{equation}
		with $\kappa$ being some bilinear map, leads to equivalent flatness conditions.
		
		Higher gauge algebras are usually modelled by $L_\infty$-algebras (generalisations of differential graded Lie algebras), which are non-trivial only in non-positive degrees. Tensoring these with the differential graded commutative algebra of differential forms, we obtain larger $L_\infty$-algebras. These $L_\infty$-algebras come with a natural gauge theory, called homotopy Maurer--Cartan theory.
		
		Whilst homotopy Maurer--Cartan theory provides a natural and suitable theory for flat higher connections with all expected properties, it typically fails for non-flat connections. In particular, gauge transformations generically only close up to curvature terms, two gauge transformations linked by a higher gauge transformation have images that differ by curvature terms, and under a quasi-isomorphism of gauge $L_\infty$-algebras, gauge transformations are only mapped to gauge transformations up to curvature terms. A comprehensive discussion of these subtleties is found in~\cite{Borsten:2024gox}. These problems are resolved for fake-flat connections, for which all curvature forms except the one of highest degree vanish. However, this condition is usually too restrictive for physical applications.
		
		An alternative approach involves identifying suitable non-flat curvature forms using either the Chern--Simons terms of~\cite{Sati:2008eg} or the more general adjustments of~\cite{Saemann:2019dsl,Kim:2019owc,Rist:2022hci,Fischer:2024vak}.  An adjustment is an additional algebraic datum on the higher gauge group or algebra that can be used to deform the naive gauge transformations suggested by homotopy Maurer--Cartan theory. These redefined gauge transformations then cure the above mentioned problems; in particular, the higher gauge transformations all close as expected.
		
		In the case of Chern--Simons theory, however, it might appear that fake-flatness is a viable approach, since the equations of motion include this condition, anyway. This would, however,  simplify the expected action functional too much. In the case of four-dimensional higher Chern--Simons theory with strict gauge 2-group, for example, one expects terms in the Lagrangian of the form $BF+\tfrac12B^2$, where $B$ is a 2-form potential and $F$ the 2-form fake curvature. Putting $F=0$ from the outset would render the theory free and much less interesting, e.g.~for any future quantisation. Generally, imposing fake-flatness from the outset on the kinematical data restricts the moduli space of connections to a zero-set of the original moduli space. As a consequence, varying the natural action functionals no longer yields full flatness of the connection.
		
		Here, we study this problem in detail and resolves it by clarifying the construction of interesting higher Chern--Simons actions that manage to satisfy the expectations listed above. In our approach, we are relying on the above mentioned adjustments.
		
		Our first key result is a surprising no-go theorem: higher gauge algebras modelled on minimal $L_\infty$-algebras\footnote{Any general $L_\infty$-algebra is quasi-isomorphic to a minimal one.} equipped with a cyclic structure (inner product) do not admit adjustments. This is significant because, as we demonstrate, quasi-isomorphic gauge $L_\infty$-algebras lead to semi-classically equivalent higher Chern--Simons theories, implying that also non-minimal $L_\infty$-algebras cannot provide interesting theories.
		
		Since a full adjustment is impossible, we take a detour: starting with an adjusted higher gauge algebra, we complete it to its natural cotangent bundle to accommodate an inner product. This construction, reminiscent of the Batalin--Vilkovisky anti-field formalism, produces what we term a \uline{half-adjusted higher gauge algebra}. As we show, all gauge and higher gauge transformations close except for the higher gauge transformations in the cotangent directions, which we explicitly exclude.
		
		We discuss the example of four-dimensional higher Chern--Simons theory in detail, presenting both local, infinitesimal descriptions as well as the integrated, finite gauge transformations. For completeness, we also develop the explicit differential cohomology describing principal 2-bundles with half-adjusted connections and their isomorphisms. We then develop the local, infinitesimal description of half-adjusted higher Chern--Simons theory in arbitrary dimensions, demonstrating that this theory has the desired properties, with a mild restriction on higher gauge transformations.
		
		We conclude by comparing alternatives to half-adjustments in general higher gauge theory. For higher Chern--Simons theories, one alternative approach introduces additional (higher) gauge symmetries to achieve the expected closure and globular structure of the higher gauge transformations. We show that for strict higher gauge algebras\footnote{Any higher gauge algebra is expected to be equivalent to a strict one by abstract nonsense.}, this leads to a trivialisation of all higher gauge transformations, contrary to the half-adjustments, thereby reducing higher gauge theory to ordinary gauge theory.
		
		Comparing our approach to other higher gauge theories in the literature, we note that for theories that are not of Chern--Simons type, the preferred procedure is to either not use a cyclic higher gauge algebra or to lose compatibility with quasi-isomorphism on the higher gauge algebra. Since neither solution is available for higher Chern--Simons theory, we believe that our approach of half-adjusted theories represents the most promising current method for deriving interesting and consistent Chern--Simons action functionals.
		
		Future work will explore higher versions of the four-dimensional Chern--Simons theory and their relations to higher-dimensional integrable models as well as the application of these theories to higher-dimensional knots.
		
		\section{\texorpdfstring{$L_\infty$}{L infinity}-algebras and higher Chern--Simons theory}
		
		We begin with a brief review of strong homotopy Lie algebras, also known as $L_\infty$-algebras, as well as the natural form of classical higher Chern--Simons theory arising from homotopy Maurer--Cartan theory. For more details and background material, we refer to the review parts in~\cite{Jurco:2018sby,Jurco:2019bvp}.
		
		\subsection{Cyclic \texorpdfstring{$L_\infty$}{L infinity}-algebras}\label{ssec:cyclic_L_infty}
		
		\paragraph{$L_\infty$-algebras.}
		\uline{$L_\infty$-algebras} generalise the notion of (differential graded) Lie algebras in that they allow for higher brackets contributing to weaker forms of the Jacobi identities. In particular, an $L_\infty$-algebra consists of a $\IZ$-graded vector space $\frL=\bigoplus_{k\in\IZ}\frL_k$ and graded anti-symmetric multi-linear maps $\mu_i:\frL\times\cdots\times\frL\rightarrow\frL$ of degree $2-i$ for all $i\in\IN^+$, called the \uline{higher products}, and which obey the \uline{homotopy Jacobi identities}
		\begin{subequations}
			\begin{equation}\label{eq:hJidentity}
				\sum_{i_1+i_2=i}\sum_{\sigma\in\overline{\rm Sh}(i_1;i)}(-1)^{i_2}\chi(\sigma;V_1,\ldots,V_i)\mu_{i_2+1}(\mu_{i_1}(V_{\sigma(1)},\ldots,V_{\sigma(i_1)}),V_{\sigma(i_1+1)},\ldots,V_{\sigma(i)})\ =\ 0
			\end{equation}
			for all homogeneous $V_{1,\ldots,i}\in\frL$ and $i\in\IN^+$. Here, the sum is over all $(i_1;i)$ \uline{unshuffles} $\sigma$, that is, permutations $\sigma$ of $\{1,\ldots,i\}$ such that $\sigma(1)<\cdots<\sigma(i_1)$ and $\sigma(i_1+1)<\cdots<\sigma(i)$. In addition, $\chi(\sigma;V_1,\ldots,V_i)$ is the \uline{Koszul sign} which is defined by
			\begin{equation}
				V_1\wedge\ldots\wedge V_i\ =\ \chi(\sigma;V_1,\ldots,V_i)V_{\sigma(1)}\wedge\ldots\wedge V_{\sigma(i)}
			\end{equation}
		\end{subequations}
		for all homogeneous $V_{1,\ldots,i}\in\frL$. 
		
		For $i=1$, the homotopy Jacobi identities~\eqref{eq:hJidentity} state that $\mu_1^2=0$, that is, $\mu_1$ is a differential. For $i=2$, we find that $\mu_1$ is a derivation with respect to the bracket $\mu_2$, and for $i=3$, we find that $\mu_2$ satisfies the Jacobi identity up to terms involving $\mu_3$.
		
		In the following, we shall denote the degree of a homogeneous element $V\in \frL$ by $|V|$. Furthermore, an $L_\infty$-algebra is called \uline{skeletal} or \uline{minimal} provided that $\mu_1=0$ and \uline{strict} provided that $\mu_{i>2}=0$. If $\frL$ is only concentrated in degrees $-n+1,\ldots,0$, that is, $\frL=\frL_{-n+1}\oplus\cdots\oplus \frL_0$, we call the $L_\infty$-algebra an \uline{$n$-term $L_\infty$-algebra}. We call an $L_\infty$-algebra with $\mu_{i>1}=0$ \uline{Abelian}.
		
		\paragraph{Example: Crossed modules of Lie algebras.}
		A simple non-trivial example of an $L_\infty$-algebra is obtained from a crossed module of Lie algebras. Recall that a \uline{crossed module of Lie algebras} consists of a pair of Lie algebras $\frh$ and $\frg$ together with a morphism $\sft:\frh\rightarrow\frg$ and an action $\acton$ of automorphisms of $\frg$ on $\frh$ such that
		\begin{equation}
			\sft(X\acton Y)\ =\ [X,\sft(Y)]
			\eand
			\sft(Y_1)\acton Y_2\ =\ [Y_1,Y_2]
		\end{equation}
		for all $X\in\frg$ and $Y,Y_{1,2}\in\frh$. The first condition is known as the \uline{equivariance condition} and the second as the \uline{Peiffer condition}. We then set $\frL=\frL_{-1}\oplus\frL_0$ with $\frL_{-1}\coloneqq\frh$ and $\frL_0\coloneqq\frg$ as well as
		\begin{equation}
			\begin{gathered}
				\mu_1(Y)\ \coloneqq\ \sft(Y)~,
				\quad
				\mu_1(X)\ \coloneqq\ 0~,
				\\
				\mu_2(X_1,X_2)\ \coloneqq\ [X_1,X_2]~,
				\quad
				\mu_2(Y_1,Y_2)\ \coloneqq\ 0~,
				\\
				\mu_2(X,Y)\ \coloneqq\ X\acton Y\ \eqqcolon\ -\mu_2(Y,X)~,
			\end{gathered}
		\end{equation}
		for all $X,X_{1,2}\in\frL_0$ and $Y,Y_{1,2}\in\frL_{-1}$. This then defines a strict $2$-term $L_\infty$-algebra. Furthermore, this construction can easily be inverted, and it thus follows that crossed modules of Lie algebras are, in fact, equivalent to strict $2$-term $L_\infty$-algebras.
		
		\paragraph{Example: String Lie 2-algebra.}
		Another very important example, which is not a differential graded Lie algebra, is the \uline{(skeletal) string Lie 2-algebra}. This is a 2-term $L_\infty$-algebra constructed from a real metric Lie algebra $(\frg,[-,-],\inner{-}{-})$. We set $\frL=\frL_{-1}\oplus\frL_0$ with $\frL_{-1}\coloneqq\IR$ and $\frL_0\coloneqq\frg$ and with the only non-vanishing higher products
		\begin{equation}
			\mu_2(X_1,X_2)\ \coloneqq\ [X_1,X_2]
			\eand
			\mu_3(X_1,X_2,X_3)\ \coloneqq\ \inner{X_1}{[X_2,X_3]}
		\end{equation}
		for all $X_{1,2,3}\in\frg$. There is also a strict version of this 2-term $L_\infty$-algebra based on path and loop spaces in $\frg$, see~\cite{Baez:2005sn}. The string Lie 2-algebra of $\frsu(2)$ is a natural higher version of $\frsu(2)\cong \frspin(3)$, and the corresponding Lie 2-group features higher analogues of many of the properties of the Lie group $\sfSU(2)$.
		
		\paragraph{Cyclic structure.}
		\uline{Cyclic $L_\infty$-algebras} generalise metric or quadratic (differential graded) Lie algebras. In particular, they are $L_\infty$-algebras that are equipped with a non-degenerate graded symmetric bilinear form\footnote{here we restrict to the real situation} $\inner{-}{-}:\frL\times \frL\rightarrow\IR$ of definite degree subject to the \uline{cyclicity condition}
		\begin{equation}
			\inner{V_1}{\mu_i(V_2,\ldots,V_{i+1})}\ =\ (-1)^{i+i(|V_1|+|V_{i+1}|)+|V_{i+1}|\sum_{j=1}^i|V_j|}\inner{V_{i+1}}{\mu_i(V_1,\ldots,V_i)}
		\end{equation}
		for all homogeneous $V_{1,\ldots,i}\in\frL$ and $i\in\IN^+$. We shall also refer to $\inner{-}{-}$ as an \uline{inner product}.
		
		\paragraph{Example: $2$-term $L_\infty$-algebras.}
		For $(\frg,[-,-],\inner{-}{-})$ a metric Lie algebra, we obtain a strict cyclic 2-term $L_\infty$-algebra $\frL=\frL_{-1}\oplus\frL_0$ by setting $\frL_{-1}\coloneqq\frg$ and $\frL_0\coloneqq\frg$ and with the only non-vanishing higher product $\mu_2(X_1,X_2)\coloneqq[X_1,X_2]$.
		
		Generally, for a (finite-dimensional) 2-term $L_\infty$-algebra $\frL=\frL_{-1}\oplus\frL_0$ to be cyclic, we need $\inner{-}{-}$ to be of degree $1$ inducing $\frL_{-1}\cong(\frL_0)^*$. For example, the string Lie 2-algebra cannot be made cyclic as, generically, $\frL_{-1}\ncong(\frL_0)^*$.
		
		\paragraph{Example: Cotangent $L_\infty$-algebra.}
		Let $(\frL,\mu_i)$ be a (finite-dimensional) $n$-term $L_\infty$-algebra. Then, we can always construct a cyclic $n$-term $L_\infty$-algebra $(\hat\frL,\hat\mu_i)$ by considering the degree-shifted cotangent space\footnote{\label{fn:degreeShift}The notation $\sfV[l]$ with $l\in\IZ$ for $\sfV=\bigoplus_k\sfV_k$ a graded vector space means the degree-shifted vector space $\sfV[l]=\bigoplus_k(\sfV[l])_k$ with $(\sfV[l])_k\coloneqq\sfV_{k+l}$. Likewise, its dual $\sfV^*$ is $\sfV^*=\bigoplus_k(\sfV^*)_k$ with $(\sfV^*)_k\coloneqq(\sfV_{-k})^*$. Hence, $(\frL^*[n-1])_k=(\frL^*)_{k+n-1}=(\frL_{1-n-k})^*$.}
		\begin{subequations}\label{eq:shifted_cotangent_construction}
			\begin{equation}
				\hat\frL\ \coloneqq\ T^*[n-1]\frL\ =\ \frL\oplus\frL^*[n-1]\ =\ \bigoplus_k(\frL_k\oplus(\frL_{1-n-k})^*)~.
			\end{equation}
			This graded vector space admits a canonical non-degenerate pairing
			\begin{equation}
				\begin{aligned}
					\inner{-}{-}\,:\,\hat\frL\times\hat\frL\ &\rightarrow\ \IR
				\end{aligned}
			\end{equation}
			of degree $n-1$ given by
			\begin{equation}
				\inner{(V_1,W^*_1)}{(V_2,W^*_2)}\ \coloneqq\ \tfrac12\big(W^*_1(V_2)+(-1)^{|V_1||W^*_2|}W^*_2(V_1)\big)
			\end{equation}
			for all homogeneous $(V_{1,2},W^*_{1,2})\in\frL\oplus\frL^*[n-1]$; this is only non-zero for $(V_1,W^*_1)\in\hat\frL_k$ and $(V_2,W^*_2)\in\hat\frL_{1-n-k}$ for all $k\in\IZ$. The higher products $\hat\mu_i$ are given by
			\begin{equation}
				\hat\mu_i\big((V_1,W^*_1),\ldots,(V_i,W^*_i)\big)\ \coloneqq\ \underbrace{\mu_i(V_1,\ldots,V_i)}_{\in\,\frL}+\underbrace{\sum_{j=1}^i\mu_i^*(V_1,\ldots,V_{j-1},W^*_j,V_{j+1},\ldots,V_i)}_{\in\,\frL^*[n-1]}
			\end{equation}
			for all $(V_{1,\ldots,i},W^*_{1,\ldots,i})\in\frL\oplus\frL^*[n-1]$ with
			\begin{equation}
				\begin{aligned}
					\mu^*_i(V_1,\ldots,V_{i-1},W^*)\ &=\ -(-1)^{|W^*||V_{i-1}|}\mu^*_i(V_1,\ldots,V_{i-2},W^*,V_{i-1})
					\\
					&\ \ \vdots
					\\
					&=\ -(-1)^{|W^*|\sum_{j=1}^{i-1}|V_j|}\mu^*_i(W^*,V_1,\ldots,V_{i-1})
				\end{aligned}
			\end{equation}
			for all $W^*\in\frL^*[n-1]$ and using the cyclification 
			\begin{equation}
				\inner{\mu^*_i(V_1,\ldots,V_{i-1},W^*)}{V_i}\ =\ -(-1)^{ i |W^*| + |W^*| \sum_{j=1}^{i-1}|V_j| }\inner{W^*}{\mu_i(V_1,\ldots, V_{i-1},V_i)}~.
			\end{equation}
		\end{subequations}
		We call the $L_\infty$-algebra $(\hat\frL,\hat\mu_i,\inner{-}{-})$ the \uline{cotangent $L_\infty$-algebra of $(\frL,\mu_i)$}, and this $L_\infty$-algebra is indeed cyclic. Furthermore, we stress that this construction is a generalisation of the construction that underlies the extension of the BRST complex to the Batalin--Vilkovisky complex in the physics literature, see also~\cite{Saemann:2019dsl} for a physical application of this construction.
		
		\paragraph{Cotangent $L_\infty$-algebra of the string Lie 2-algebra.} In the case of the string Lie 2-algebra of a metric Lie algebra $\frg$ introduced above, we have
		\begin{subequations}
			\begin{equation}
				\hat\frL\ =\ \big(\hat\frL_{-1}\xrightarrow{~\hat\mu_1~}\hat\frL_0\big)\ =\
				\left(
				\begin{array}{c}
					\frg^* \\[-0.1cm] \oplus \\[-0.1cm] \IR
				\end{array}
				\xrightarrow{~0~}
				\begin{array}{c}
					\IR^* \\[-0.1cm] \oplus \\[-0.1cm] \frg
				\end{array}
				\right)
			\end{equation}
			with $\hat\mu_2$ and $\hat\mu_3$ as the only non-vanishing higher products. The only non-vanishing dual products are\footnote{Recall a left group action $\acton\,:\sfG\times\sfV\rightarrow\sfV$ of a group $\sfG$ on a vector space $\sfV$ induces an action on the dual space $\sfV^*$ defined by $(g\acton v^*)(v)=v^*(g^{-1}\acton v)$ for all $v^*\in\sfV^*$ and $v\in\sfV$, implying our sign convention here.}
			\begin{equation}
				\begin{aligned}
					\mu^*_2(X_1,X^*)(X_2)\ &=\ -X^*([X_1,X_2])~,
					\\
					\mu^*_3(X_1,X_2,Y^*)(X_3)\ &=\ Y^*(\inner{[X_1,X_2]}{X_3})
				\end{aligned}
			\end{equation}
		\end{subequations}
		for all $X_{1,2,3}\in\frg$, $X^*\in\frg^*$, and $Y^*\in\IR^*$.
		
		\subsection{Higher Chern--Simons theory}\label{ssec:higher_CS}
		
		\paragraph{$L_\infty$-algebra valued differential forms.}
		It is well-known that the tensor product of a commutative algebra and a Lie algebra naturally forms a Lie algebra. This observation extends to the tensor product of a differential graded commutative algebra and an $L_\infty$-algebra naturally forming an $L_\infty$-algebra. Let us explore this for the example of the de~Rham complex and an $L_\infty$-algebra.
		
		In particular, let $M$ be a manifold and consider the \uline{de~Rham complex} $\Omega^\bullet(M)$ of differential forms on $M$ with $\rmd$ the exterior derivative and $\wedge$ the standard wedge product between differential forms. This is a differential graded algebra. Hence, given an $L_\infty$-algebra $(\frL,\mu_i)$, we can form a new $L_\infty$-algebra $\big(\Omega^\bullet(M,\frL),\mu^{\Omega^\bullet(M,\frL)}_i\big)$ by means of
		\begin{subequations}
			\begin{equation}
				\Omega^\bullet(M,\frL)\ \coloneqq\ \bigoplus_{k\in\IZ}\Omega^\bullet_k(M,\frL)
				\ewith
				\Omega^\bullet_k(M,\frL)\ \coloneqq\ \bigoplus_{i+j=k}\Omega^i(M)\otimes\frL_j
			\end{equation}
			and
			\begin{equation}\label{eq:tensored_higher_products}
				\begin{aligned}
					\mu_1^{\Omega^\bullet(M,\frL)}(\omega\otimes V)\ &\coloneqq\ \rmd\omega\otimes V+(-1)^{|\omega|}\omega\otimes\mu_1(V)~,
					\\[5pt]
					\mu_i^{\Omega^\bullet(M,\frL)}(\omega_1\otimes V_1,\ldots,\omega_i\otimes V_i)\ &\coloneqq\ (-1)^{i\sum_{j=1}^i|\omega_i|+\sum_{j=0}^{i-2}|\omega_{i-j}|\sum_{k=1}^{i-j-1}|V_k|}
					\\
					&\kern1cm\times(\omega_1\wedge\ldots\wedge\omega_i)\otimes\mu_i(V_1,\ldots V_i)
				\end{aligned}
			\end{equation}
		\end{subequations}
		for all homogeneous $\omega,\omega_{1,\ldots,i}\in\Omega^\bullet(M)$ and $V,V_{1,\ldots,i}\in \frL$.
		
		Suppose now that the $L_\infty$-algebra is a cyclic $(d-2)$-term $L_\infty$-algebra with inner product $\inner{-}{-}_\frL:\frL\times \frL\rightarrow\IR$ and $M$ a $d$-dimensional compact oriented manifold without boundary. Then, $\Omega^\bullet(M,\frL)$ can be endowed with the inner product
		\begin{equation}
			\inner{\omega_1\otimes V_1}{\omega_2\otimes V_2}^{\Omega^\bullet(M,\frL)}\ \coloneqq\ (-1)^{|\omega_2||V_1|}\int_M\omega_1\wedge\omega_2\,\inner{V_1}{V_2}
		\end{equation}
		for all homogeneous $\omega_{1,2}\in\Omega^\bullet(M)$ and $V_{1,2}\in\frL$.
		
		\paragraph{Higher Chern--Simons theory.}
		As above, consider a $d$-dimensional compact oriented manifold $M$ without boundary. Furthermore, let $(\frL,\mu_i,\inner{-}{-})$ be a cyclic $(d-2)$-term $L_\infty$-algebra. We shall refer to this $L_\infty$-algebra as the gauge $L_\infty$-algebra. Then, for $a\in\Omega^\bullet_1(M,\frL)$, the action of \uline{higher Chern--Simons theory} is defined as\footnote{This is simply the homotopy Maurer--Cartan action for the $L_\infty$-algebra $\Omega^\bullet(M,\frL)$, see again e.g.~\cite{Jurco:2018sby,Jurco:2019bvp} for a review.}
		\begin{equation}\label{eq:higherCSAction}
			S_d\ \coloneqq\ \sum_i\frac1{(i+1)!}\inner{a}{\mu_i^{\Omega^\bullet(M,\frL)}(a,\ldots,a)}^{\Omega^\bullet(M,\frL)}~.
		\end{equation}
		As a consequence of the homotopy Jacobi identities and the cyclicity of the inner product, this action is invariant under the infinitesimal gauge transformations
		\begin{equation}\label{eq:gaugeTransformationsHigherCSTheory}
			\delta_{c_0}a\ \coloneqq\ \sum_i\frac1{i!}\mu_{i+1}^{\Omega^\bullet(M,\frL)}(a,\ldots,a,c_0)
		\end{equation}
		for all $c_0\in\Omega^\bullet_0(M,\frL)$. We also have higher gauge transformations that are defined recursively by
		\begin{equation}\label{eq:higherGaugeTransformationsHigherCSTheory}
			\delta_{c_{-k-1}}c_{-k}\ \coloneqq\ \sum_i\frac1{i!}\mu_{i+1}^{\Omega^\bullet(M,\frL)}(a,\ldots,a,c_{-k-1})
		\end{equation}
		for all $c_{-k}\in\Omega^\bullet_{-k}(M,\frL)$. Furthermore, the equation of motion following from~\eqref{eq:higherCSAction} is
		\begin{equation}\label{eq:eomHigherCSTheory}
			f\ \coloneqq\ \sum_i\frac1{i!}\mu_i^{\Omega^\bullet(M,\frL)}(a,\ldots,a)\ =\ 0~.
		\end{equation}
		We call $a\in\Omega^\bullet_1(M,\frL)$ the \uline{gauge potential} and $f\in\Omega^\bullet_2(M,\frL)$ its \uline{curvature}. Solutions $a$ to the equation $f=0$ are called \uline{Maurer--Cartan elements}. Note that the curvature $f$ transforms under the gauge transformations~\eqref{eq:gaugeTransformationsHigherCSTheory} as
		\begin{equation}\label{eq:gaugeTransformationCurvatureHigherCSTheory}
			\delta_{c_0}f\ =\ \sum_i\frac1{i!}\mu_{i+2}^{\Omega^\bullet(M,\frL)}(a,\ldots,a,f,c_0)~.
		\end{equation}
		
		\paragraph{Example: $3$-dimensional Chern--Simons theory.}
		For $d=3$, we have $\frL=\frL_0\eqqcolon\frg$ with only $\mu_2(-,-)\eqqcolon[-,-]$ non-trivial. Hence, the gauge $L_\infty$-algebra reduces to a metric Lie algebra. Consequently, $a\eqqcolon A\in\Omega^\bullet_1(M,\frL)=\Omega^1(M)\otimes\frg$ and the higher Chern--Simons action~\eqref{eq:higherCSAction} reduces to the action of standard Chern--Simons theory,
		\begin{equation}
			S_{d=3}\ =\ \int_M\big\{\tfrac12\inner{A}{\rmd A}+\tfrac1{3!}\inner{A}{[A,A]}\big\}\,.
		\end{equation}
		In this case, the gauge transformations~\eqref{eq:gaugeTransformationsHigherCSTheory} are the standard ones,
		\begin{equation}
			\delta_\alpha A\ =\ \rmd \alpha+[A,\alpha]
		\end{equation}
		with $c_0\eqqcolon \alpha\in\Omega^\bullet_0(M,\frL)=\Omega^0(M)\otimes\frg$, and there are no higher gauge transformations. Furthermore, the equation of motion~\eqref{eq:eomHigherCSTheory} is simply
		\begin{equation}
			F\ =\ \rmd A+\tfrac12[A,A]\ =\ 0~,
		\end{equation}
		where $f\eqqcolon F\in\Omega^\bullet_2(M,\frL)=\Omega^2(M)\otimes\frg$ is the standard curvature $2$-form. Finally,~\eqref{eq:gaugeTransformationCurvatureHigherCSTheory} reduces to
		\begin{equation}
			\delta_\alpha F\ =\ [F,\alpha]~.
		\end{equation}
		
		\paragraph{Example: $4$-dimensional Chern--Simons theory.}
		For $d=4$, we have $\frL=\frL_{-1}\oplus \frL_0$. Consequently, $a\eqqcolon(A,B)\in\Omega^\bullet_1(M,\frL)=\Omega^1(M)\otimes\frL_0\oplus\Omega^2(M)\otimes\frL_{-1}$. In this case, the higher Chern--Simons action~\eqref{eq:higherCSAction} becomes
		\begin{equation}\label{eq:4dCSAction}
			S_{d=4}\ =\ \int_M\big\{\inner{B}{\rmd A+\tfrac12\mu_2(A,A)+\tfrac12\mu_1(B)}_\frL-\tfrac1{4!}\inner{A}{\mu_3(A,A,A)}\big\}\,.
		\end{equation}
		Here and in the following, our notation is such that the higher products $\mu_i$ act on the $L_\infty$-part of a differential form. For example,
		\begin{equation}
			\mu_3(A,A,A)\ =\ -\mu_3^{\Omega^\bullet(M,\frL)}(A,A,A)~,
		\end{equation}
		because the 3-form part of $\mu_3^{\Omega^\bullet(M,\frL)}(A,A,A)$ has been moved to the left of the higher product $\mu_3$ according to~\eqref{eq:tensored_higher_products}.
		
		The gauge transformations~\eqref{eq:gaugeTransformationsHigherCSTheory} read as
		\begin{equation}\label{eq:4dGaugeTransformations}
			\begin{aligned}
				\delta_{(\alpha,\lambda)}A\ &=\ \rmd \alpha+\mu_2(A,\alpha)-\mu_1(\lambda)~,
				\\
				\delta_{(\alpha,\lambda)}B\ &=\ \mu_2(B,\alpha)+\rmd\lambda+\mu_2(A,\lambda)+\tfrac12\mu_3(A,A,\alpha)
			\end{aligned}
		\end{equation}
		with $c_0\eqqcolon(\alpha,\lambda)\in\Omega^\bullet_0(M,\frL)=\Omega^0(M)\otimes\frL_0\oplus\Omega^1(M)\otimes\frL_{-1}$. We now also have higher gauge transformations
		\begin{equation}\label{eq:4dHigherGaugeTransformation}
			\begin{aligned}
				\delta_\sigma \alpha\ &=\ \mu_1(\sigma)~,
				\\
				\delta_\sigma\lambda\ &=\ \rmd\sigma+\mu_2(A,\sigma)
			\end{aligned}
		\end{equation}
		with $c_{-1}\eqqcolon\sigma\in\Omega^\bullet_{-1}(M,\frL)=\Omega^0(M)\otimes\frL_{-1}$. Furthermore, the equation of motion~\eqref{eq:eomHigherCSTheory} decomposes as
		\begin{equation}\label{eq:4d_curvatures}
			\begin{aligned}
				F\ \coloneqq\ \rmd A+\tfrac12\mu_2(A,A)+\mu_1(B)\ &=\ 0~,
				\\
				H\ \coloneqq\ \rmd B+\mu_2(A,B)-\tfrac1{3!}\mu_3(A,A,A)\ &=\ 0
			\end{aligned}
		\end{equation}
		with $f\eqqcolon(F,H)\in\Omega^\bullet_2(M,\frL)=\Omega^2(M)\otimes\frL_0\oplus\Omega^3(M)\otimes\frL_{-1}$. The gauge transformations of these curvature forms are given by a reduction of~\eqref{eq:gaugeTransformationCurvatureHigherCSTheory} to
		\begin{equation}
			\begin{aligned}
				\delta_{(\alpha,\lambda)}F\ &=\ \mu_2(F,\alpha)~,
				\\
				\delta_{(\alpha,\lambda)}H\ &=\ \mu_2(H,\alpha)+\mu_2(F,\lambda)+\mu_3(A,F,\alpha)~,
			\end{aligned}
		\end{equation}
		and the Bianchi identities read as
		\begin{equation}\label{eq:4d_Bianchi}
			\begin{aligned}
				\rmd F+\mu_2(A,F)-\mu_1(H)\ &=\ 0~,
				\\
				\rmd H+\mu_2(A,H)-\mu_2(F,B)+\tfrac12\mu_3(A,A,F)\ &=\ 0~.
			\end{aligned}
		\end{equation}
		
		It is now not too difficult to see that
		\begin{subequations}
			\begin{equation}\label{eq:gauge_trafo_not_closing}
				\begin{aligned}
					[\delta_{(\alpha,\lambda)},\delta_{(\alpha',\lambda')}]A\ &=\ \delta_{(\alpha'',\lambda'')}A~,
					\\
					[\delta_{(\alpha,\lambda)},\delta_{(\alpha',\lambda')}]B\ &=\ \delta_{(\alpha'',\lambda'')}B+\mu_3(F,\alpha,\alpha')
				\end{aligned}
			\end{equation}
			for any two gauge transformations where
			\begin{equation}
				\begin{aligned}
					\alpha''\ &\coloneqq\ \mu_2(\alpha',\alpha)~,
					\\
					\lambda''\ &\coloneqq\ \mu_2(\alpha',\lambda)-\mu_2(\alpha,\lambda')-\mu_3(A,\alpha',\alpha)~.
				\end{aligned}
			\end{equation}
			We also note that
			\begin{equation}\label{eq:higher_gauge_trafo_not_closing}
				\delta_\sigma(\delta_{\alpha,\lambda}B)\ =\ \mu_2(F,\sigma)
			\end{equation}
			for $\sigma\in\Omega^0(M)\otimes\frL_{-1}$; see~\cite[Appendix C]{Jurco:2018sby} for the explicit computation in the general case.
		\end{subequations}
		
		\subsection{Equivalences of gauge \texorpdfstring{$L_\infty$}{L infinity}-algebras imply semi-classical equivalence}\label{ssec:equivalences}
		
		Homotopy algebras come with a notion of equivalence, called \uline{quasi-isomorphism}, which replaces the usual isomorphism of Lie algebras. However, many field theories do not respect this equivalence relation, contrary to expectations. That is, there are field theories where replacing the gauge $L_\infty$-algebra with a quasi-isomorphic one yields a physically inequivalent field theory, see e.g.~the discussion of the higher Stueckelberg model in~\cite{Borsten:2024gox}\footnote{Another, physically relevant case is that of the gauge $L_\infty$-algebra underlying the six-dimensional $\caN=(1,0)$ model of~\cite{Saemann:2017zpd,Saemann:2019dsl,Rist:2020uaa}, which is quasi-isomorphic to an ordinary Lie algebra.}. For higher Chern--Simons theories, however, the quasi-isomorphism of gauge $L_\infty$-algebras indeed induces an equivalence of field theories, which we may use to simplify our discussion. We briefly develop the relevant background in the following.
		
		\paragraph{Chevalley--Eilenberg algebra of an $L_\infty$-algebra.}
		Morphisms of $L_\infty$-algebras are most easily understood in a dual formulation in terms of their Chevalley--Eilenberg algebras. Recall that the \uline{Chevalley--Eilenberg algebra} of an $L_\infty$-algebra $\frL$, denoted by $\sfCE(\frL)$, is the symmetric tensor product\footnote{Technically, this produces a curved $L_\infty$-algebra, and one should restrict to the reduced symmetric tensor product, see e.g.~\cite{Jurco:2018sby} for a detailed explanation.} algebra $\bigodot^\bullet(\frL[1])^*$ together with a differential $\sfd_{\sfCE}$, which is a derivation of $\odot$ and encodes the higher products~\cite{Sati:2008eg}.\footnote{From \cref{fn:degreeShift} on page \cpageref{fn:degreeShift}, we obtain for $(\frL[1])^*=\bigoplus_k((\frL[1])^*)_k$ that $((\frL[1])^*)_k=((\frL[1])_{-k})^*=(\frL_{1-k})^*$.}
		
		We briefly illustrate the construction for a 2-term $L_\infty$-algebra $\frL=\frL_{-1}\oplus\frL_0$. Here, $(\frL[1])^*$ is concentrated in degrees $k=1,2$ with
		\begin{subequations}
			\begin{equation}
				((\frL[1])^*)_k\ =\
				\begin{cases}
					(\frL_0)^* & \efor k\ =\ 1
					\\
					(\frL_{-1})^* & \efor k\ =\ 2
				\end{cases}~.
			\end{equation}
			Consequently, the Chevalley--Eilenberg algebra $\sfCE(\frL)$ is generated by elements $\tte^\alpha\in\sfCE(\frL)$ of degree $1$ and elements $\tte^a\in\sfCE(\frL)$ of degree $2$ with $\alpha,\beta,\ldots=1,\ldots,\dim(\frL_0)$ and $a,b,\ldots=1,\ldots,\dim(\frL_{-1})$. The differential is defined by its action on these generators, which is necessarily of the form
			\begin{equation}\label{eq:2TermCEDifferential}
				\begin{aligned}
					\sfd_\sfCE\tte^\alpha\ &=\ -\tfrac12f_{\beta\gamma}{}^\alpha\tte^\beta\tte^\gamma-f_a{}^\alpha\tte^a~,
					\\
					\sfd_\sfCE\tte^a\ &=\ -f_{\alpha b}{}^a\tte^\alpha\tte^b+\tfrac1{3!}f_{\alpha\beta\gamma}{}^a\tte^\alpha\tte^\beta\tte^\gamma
				\end{aligned}
			\end{equation}
		\end{subequations}
		with $f_{\beta\gamma}{}^\alpha$, $f_a{}^\alpha$, $f_{\alpha b}{}^a$, and, $f_{\alpha\beta\gamma}{}^a$ the \uline{structure constants}. The latter define the higher products
		\begin{equation}\label{eq:2TermHigherProductsFromCE}
			\begin{gathered}
				\mu_1(\ttE_a)\ =\ f_a{}^\alpha\ttE_\alpha~,
				\\
				\mu_2(\ttE_\alpha,\ttE_\beta)\ =\ f_{\alpha\beta}{}^\gamma\ttE_\gamma~,
				\quad
				\mu_2(\ttE_\alpha,\ttE_a)\ =\ f_{\alpha a}{}^b\ttE_b~,
				\\
				\mu_3(\ttE_\alpha,\ttE_\beta,\ttE_\gamma)\ =\ f_{\alpha\beta\gamma}{}^a\ttE_a~,
			\end{gathered}
		\end{equation}
		where $\ttE_\alpha$ and $\ttE_a$ are basis vectors of $\frL_0$ and $\frL_{-1}$, dual to $\tte^\alpha$ and $\tte^a$, respectively. For more details, we refer to~\cite{Jurco:2018sby}.
		
		\paragraph{Morphisms.}
		A \uline{morphism of $L_\infty$-algebras} $\phi:\frL\rightarrow\tilde\frL$ is dually a morphism of differential graded commutative algebras between the corresponding Chevalley--Eilenberg algebras $\phi^*:\sfCE(\tilde\frL)\rightarrow\sfCE(\frL)$. Unpacking this definition, we note that a morphism $\phi$ consists of graded anti-symmetric, $i$-ary multilinear maps $\phi_i:\frL\times\cdots\times\frL\rightarrow\tilde\frL$ of degree $1-i$ that intertwine between the higher products:
		\begin{equation}\label{eq:intertwining_relations}
			\begin{aligned}
				&\tilde\mu_1(\phi_1(V_1))\ =\ \phi_1\big(\mu_1(V_1)\big)~,
				\\
				&\tilde\mu_2(\phi_1(V_1),\phi_1(V_2))\ =\ \phi_1\big(\mu_2(V_1,V_2)\big)-\phi_2\big(\mu_1(V_1),V_2\big)+(-1)^{|V_1|\,|V_2|}\phi_2\big(\mu_1(V_2),V_1\big)
				\\
				&\hspace{1cm}-\tilde\mu_1(\phi_2(V_1,V_2))~,
				\\
				&\tilde\mu_3(\phi_1(V_1),\phi_1(V_2),\phi_1(V_3))\ =\ (-1)^{|V_2|\,|V_3|}\phi_2\big(\mu_2(V_1,V_3),V_2\big)
				\\
				&\hspace{1cm}+(-1)^{|V_1|(|V_2|+|V_3|)+1}\phi_2\big(\mu_2(V_2,V_3),V_1)+(-1)^{|V_1|\,|V_2|+1}\phi_3\big(\mu_1(V_2),V_1,V_3\big)
				\\
				&\hspace{1cm}+(-1)^{(|V_1|+|V_2|)|V_3|} \phi_3\big(\mu_1(V_3),V_1,V_2\big)+(-1)^{|V_1|}\tilde\mu_2(\phi_1(V_1),\phi_2(V_2,V_3))
				\\
				&\hspace{1cm}
				-(-1)^{(|V_1|+1)|V_2|}\tilde\mu_2(\phi_1(V_2),\phi_2(V_1,V_3))+(-1)^{(|V_1|+|V_2|+1)|V_3|}\tilde\mu_2(\phi_1(V_3),\phi_2(V_1,V_2))
				\\
				&\hspace{1cm}
				+\phi_1\big(\mu_3(V_1,V_2,V_3)\big)-\phi_2\big(\mu_2(V_1,V_2),V_3\big)+\phi_3\big(\mu_1(V_1),V_2,V_3\big)
				\\
				&\hspace{1cm}-\tilde\mu_1(\phi_3(V_1,V_2,V_3))
				\\
				&\hspace{4cm}\vdots
			\end{aligned}
		\end{equation}
		for all homogeneous $V_{1,2,3}\in\frL$. If $\phi_{i>1}=0$, we call the morphism $\phi$ \uline{strict}.
		
		\paragraph{Minimal model theorem.}
		We note that the component $\phi_1$ in a morphism of $L_\infty$-algebras $\phi:\frL\rightarrow\tilde\frL$ is a cochain map between the cochain complexes contained in $\frL$ and $\tilde\frL$, and hence descends to the corresponding cohomologies. If $\phi_1$ induces an isomorphism between these cohomologies, we call $\phi$ a \uline{quasi-isomorphism}, extending the corresponding notion from cochain complexes.
		
		There are specialisations to cyclic morphisms, which we suppress here; for more details, see again the review in~\cite{Jurco:2018sby,Jurco:2019bvp}.
		
		By the minimal model theorem~\cite{kadeishvili1982algebraic,Kajiura:2003ax}, any $L_\infty$-algebra $\frL$ induces an $L_\infty$-algebra structure $\frL^\circ$ on its cohomology such that $\frL$ and $\frL^\circ$ are quasi-isomorphic. The $L_\infty$-algebra $\frL^\circ$ is called a \uline{minimal model} for $\frL$, and it is unique up to an $L_\infty$-algebra isomorphisms, i.e.~up to a morphism of $L_\infty$-algebras with $\phi_1$ strictly invertible.
		
		\paragraph{Semi-classical equivalence of field theories.}
		As remarked above, the Batalin--Vilkovisky complex of a field theory is the Chevalley--Eilenberg algebra of an $L_\infty$-algebra, and a minimal model of this $L_\infty$-algebra describes the tree-level S-matrix of the field theory in a particular basis\footnote{Moreover, the computation of this minimal model via homological perturbation theory corresponds precisely to the tree-level Feynman diagram expansion.}, see e.g.~\cite{Jurco:2018sby}.
		
		Two field theories whose Batalin--Vilkovisky complexes correspond to quasi-isomorphic $L_\infty$-algebras have isomorphic minimal models, and hence equivalent S-matrices. Such theories are called \uline{semi-classically equivalent}, and this is the appropriate notion of equivalence of classical field theories.
		
		\paragraph{Equivalence of higher Chern--Simons theories.}
		Consider a $d$-dimensional compact oriented manifold $M$ without boundary together with two quasi-isomorphic $(d-2)$-term $L_\infty$-algebras $\frL$ and $\tilde\frL$ with $\phi:\frL\rightarrow\tilde\frL$ a quasi-isomorphism. We note that $\phi$ induces a quasi-isomorphism $\hat\phi$ between $\Omega^\bullet(M,\frL)$ and $\Omega^\bullet(M,\tilde\frL)$ by
		\begin{equation}
			\hat\phi_i(\omega_1\otimes V_1,\ldots,\omega_i\otimes V_i)\ \coloneqq\ (-1)^{(i+1)\sum_{j=1}^{i}|\omega_j|}\omega_1\wedge\ldots\wedge\omega_i\otimes\phi_i(V_1,\ldots,V_i)
		\end{equation}
		for all homogeneous $\omega_{1,\ldots,i}\in\Omega^\bullet(M)$ and $V_{1,\ldots,i}\in\frL$. It is easily seen that if the $\phi_i$ satisfy the appropriate intertwining relations~\eqref{eq:intertwining_relations} between $\frL$ and $\tilde\frL$, then the $\hat\phi_i$ satisfy the appropriate relations between $\Omega^\bullet(M,\frL)$ and $\Omega^\bullet(M,\tilde\frL)$.\footnote{Much more abstractly, tensor products generically respect chain homotopy equivalences.} We thus obtain the following theorem.
		
		\begin{theorem}\label{thm:equivalence}
			Quasi-isomorphic $(d-2)$-term $L_\infty$-algebras yield semi-classically equivalent $d$-dimensional higher Chern--Simons theories.
		\end{theorem}
		
		This theorem has an important corollary.
		
		\begin{corollary}\label{cor:triviality}
			Consider a $(d-2)$-term gauge $L_\infty$-algebra which is quasi-isomorphic to the trivial $L_\infty$-algebra\footnote{Here, trivial means the zero cochain complex. That is, the considered $L_\infty$-algebra have trivial cohomology.}. Then, the corresponding $d$-dimensional higher Chern--Simons theory is semi-classically equivalent to the trivial theory.
		\end{corollary}
		
		In other words, whilst there are non-trivial field theories with gauge $L_\infty$-algebras with trivial cohomology (the higher Stueckelberg model discussed in~\cite{Borsten:2024gox} is a very simple such an example), higher Chern--Simons theories with such a gauge $L_\infty$-algebras are always semi-classically trivialisable.
		
		\subsection{Problems with this formulation}\label{ssec:problems}
		
		In the example of four-dimensional Chern--Simons theory from \cref{ssec:higher_CS}, we can see two features of the kinematical data that are undesirable:
		\begin{enumerate}[(i)]
			\itemsep-1pt
			\item Equation~\eqref{eq:gauge_trafo_not_closing} shows that the commutator of two gauge transformations is not a gauge transformation unless $\mu_3(F,\alpha',\alpha)=0$ for all $\alpha,\alpha'\in \Omega^0(M)\otimes\frL_0$.
			\item Equation~\eqref{eq:higher_gauge_trafo_not_closing} shows that gauge transformations related by a higher gauge transformations have different image unless $\mu_2(F,\sigma)=0$ for all $\sigma\in\Omega^0(M)\otimes\frL_{-1}$.
		\end{enumerate}
		Both problems signal that the corresponding BRST-complex, encoding the kinematic data together with gauge and higher gauge transformations is open and for generic 2-term $L_\infty$-algebras only closes for $F=0$. This condition is the so-called \uline{fake-flatness constraint}, which is ubiquitous in the literature on higher principal bundles. It has, however, significant problems in physical applications, see the discussion in~\cite{Borsten:2024gox}. As explained in the introduction, the argument that in the case of higher Chern--Simons theory, $F=0$ is anyway an equation of motion, so one could restrict the kinematical data to those $a=A+B$ with $F=0$ is unsatisfactory.
		
		These problems rather evidently persist in higher-dimensional Chern--Simons theories. In particular, the analogue of the fake curvature condition is that all higher curvature forms have to vanish except for the top component.
		
		The above problems were encountered and discussed before in other forms of higher gauge theories, and there is a common solution to these issues: to modify the definition of curvature, which goes hand-in-hand with a modification of the definition of (higher) gauge transformations. Particular types of such modifications were called \uline{Chern--Simons terms} in~\cite{Sati:2008eg}, and the general suitable modifications were called an \uline{adjustment} in~\cite{Saemann:2019dsl}. We discuss adjustments extensively in the following section.
		
		\section{Obstruction to fully adjusted higher Chern--Simons theory}
		
		Adjustments of $L_\infty$-algebras address the issues with the higher connections and their gauge symmetries introduced in \cref{ssec:higher_CS}. Essentially, they ensure that the $L_\infty$-algebra corresponding to a Batalin--Vilkovisky-action can be consistently truncated to an $L_\infty$-algebra in degrees $\leq 1$, describing the kinematic data (degree $1$) and the higher Lie algebra of higher gauge transformations (degrees $\leq 0$) together with its action on the kinematic data.
		
		As we show in the following, however, there are no adjusted $L_\infty$-algebras that would be suitable for a construction of non-trivial higher Chern--Simons theory.
		
		\subsection{Unadjusted and adjusted connections}\label{ssec:un_adj_connections}
		
		Let $M$ be a $d$-dimensional compact oriented manifold without boundary with de~Rham complex $\Omega^\bullet(M)$.
		
		\paragraph{Connections as morphisms.}
		Flat connections taking values in an $L_\infty$-algebra $\frL$ can be understood as morphisms of differential graded commutative algebras $\sfCE(\frL)\rightarrow\Omega^\bullet(M)$. In order to describe general connections, we introduce the \uline{Weil algebra} $\sfW(\frL)$ of an $L_\infty$-algebra $\frL$, which is the Chevalley--Eilenberg algebra $\sfCE(T[1]\frL)$ of the \uline{inner derivation $L_\infty$-algebra} $T[1]\frL=\frL\oplus\frL[1]$. Concretely,
		\begin{subequations}
			\begin{equation}
				\sfW(\frL)\ \coloneqq\ \big(\bigodot{}^\bullet((\frL[1])^*\oplus(\frL[2])^*),\sfd_\sfW\big)~,
			\end{equation}
			and the differential $\sfd_\sfW$ is defined by
			\begin{equation}\label{eq:WeilDifferential}
				\sfd_\sfW|_{(\frL[1])^*}\ \coloneqq\ \sfd_\sfCE+\sigma
				\eand
				\sfd_\sfW\circ\sigma+\sigma\circ\sfd_\sfCE\ =\ 0
			\end{equation}
		\end{subequations}
		with $\sfd_\sfCE$ the differential in $\sfCE(\frL)$, encoding the higher products of $\frL$, and with $\sigma:(\frL[1])^*\rightarrow(\frL[2])^*$ the evident shift isomorphism~\cite{Sati:2008eg}. Moreover, there is a canonical projection
		\begin{equation}\label{eq:projection}
			\sfW(\frL)\ \rightarrow\ \sfCE(\frL)~.
		\end{equation}
		General connections are then morphisms of differential graded commutative algebras
		\begin{equation}\label{eq:morphism_Weil}
			\caA\,:\,\sfW(\frL)\ \rightarrow\ \Omega^\bullet(M)~.
		\end{equation}
		
		\paragraph{Example: 2-term $L_\infty$-algebras.}
		As an instructive example, consider a 2-term $L_\infty$-algebra $\frL=\frL_{-1}\oplus\frL_0$. Then, $(\frL[1])^*\oplus(\frL[2])^*$ is concentrated in degrees $k=1,2,3$ with
		\begin{subequations}\label{eq:Weil_2_term}
			\begin{equation}
				((\frL[1])^*\oplus(\frL[2])^*)_k\ =\
				\begin{cases}
					(\frL_0)^* & \efor k\ =\ 1
					\\
					(\frL_{-1})^*\oplus(\frL_0)^* & \efor k\ =\ 2
					\\
					(\frL_{-1})^* & \efor k\ =\ 3
				\end{cases}~.
			\end{equation}
			Consequently, the Weil algebra $\sfW(\frL)$ is generated by elements $\tte^\alpha\in\sfW(\frL)$ of degree $1$, elements $\tte^a\in\sfW(\frL)$ and $\hat\tte^\alpha\coloneqq\sigma(\tte^\alpha)\in\sfW(\frL)$ of degree $2$, and elements $\hat\tte^a\coloneqq\sigma(\tte^a)$ of degree $3$ with $\alpha,\beta,\ldots=1,\ldots,\dim(\frL_0)$ and $a,b,\ldots=1,\ldots,\dim(\frL_{-1})$. The action of the differential on these generators is then
			\begin{equation}\label{eq:2TermWeilDifferential}
				\begin{aligned}
					\sfd_\sfW\tte^\alpha\ &\ =-\tfrac12f_{\beta\gamma}{}^\alpha\tte^\beta\tte^\gamma-f_a{}^\alpha\tte^a+\hat\tte^\alpha~,
					\\
					\sfd_\sfW \tte^a\ &=\ -f_{\alpha b}{}^a\tte^\alpha\tte^b+\tfrac1{3!}f_{\alpha\beta\gamma}{}^a\tte^\alpha\tte^\beta\tte^\gamma+\hat\tte^a~,
					\\
					\sfd_\sfW\hat\tte^\alpha\ &=\ -f_{\beta\gamma}{}^\alpha\tte^\beta\hat\tte^\gamma+f_a{}^\alpha\hat\tte^a~,
					\\
					\sfd_\sfW\hat\tte^a\ &=\ -\tfrac12f_{\alpha\beta\gamma}{}^a\tte^\alpha\tte^\beta\hat\tte^\gamma+f_{\alpha b}{}^a\hat\tte^\alpha \tte^b-f_{\alpha b}{}^a\tte^\alpha\hat\tte^b~.
				\end{aligned}
			\end{equation}
		\end{subequations}
		with $f_{\beta\gamma}{}^\alpha$, $f_a{}^\alpha$, $f_{\alpha b}{}^a$, and, $f_{\alpha\beta\gamma}{}^a$ the structure constants of $\sfd_\sfCE$ given in~\eqref{eq:2TermCEDifferential}. A morphism from $\sfW(\frL)$ to $\Omega^\bullet(M)$ is then given by an assignment
		\begin{equation}
			\begin{gathered}
				\tte^\alpha\ \mapsto\ A^\alpha\ \in\ \Omega^1(M)~,~~~
				\tte^a\ \mapsto\ B^a\ \in\ \Omega^2(M)~,
				\\
				\hat\tte^\alpha\ \mapsto\ F^\alpha\ \in\ \Omega^2(M)
				\ewith
				F^\alpha\ =\ \rmd A^\alpha+\tfrac12f_{\beta\gamma}{}^\alpha A^\beta\wedge A+f_a{}^\alpha B^a\ \in\ \Omega^2(M)~,
				\\
				\hat\tte^a\ \mapsto\ H^a\ \in\ \Omega^3(M)
				\ewith
				H^a\ =\ \rmd B^a+f_{\alpha b}{}^aA^\alpha\wedge B^b-\tfrac1{3!}f_{\alpha\beta\gamma}{}^aA^\alpha\wedge A^\beta\wedge A^\gamma~.
			\end{gathered}
		\end{equation}
		This is~\eqref{eq:4d_curvatures} when spelled out using the basis vectors $\ttE_\alpha\in\frL_0$ and $\ttE_a\in\frL_{-1}$, that is, $A=A^\alpha\ttE_\alpha$ and $B=B^a\ttE_a$ as well as~\eqref{eq:2TermHigherProductsFromCE}. Note that the expressions of the degree $2$ and $3$ elements $F^\alpha$ and $H^a$ follow because of the first two equations in~\eqref{eq:2TermWeilDifferential}. Likewise, the Bianchi identities~\eqref{eq:4d_Bianchi} can be recovered from the last two equations of~\eqref{eq:2TermWeilDifferential}.
		
		\paragraph{Adjustments.}
		As observed in~\cite{Saemann:2019dsl}, the problems with the gauge structure can be resolved by applying an automorphism on the Weil algebra\footnote{In string-like $L_\infty$-algebras, this is the same as a Chern--Simons term, as introduced in~\cite{Sati:2008eg,Fiorenza:2011jr}. See also~\cite{Borsten:2024gox} and references therein for further details.} which preserves the image of the projection~\eqref{eq:projection}. The functions parametrising suitable such automorphisms are called \uline{adjustments}~\cite{Saemann:2019dsl,Kim:2019owc,Borsten:2021ljb,Rist:2022hci}.
		
		Concretely, the automorphism changes the generators $\hat\tte^A$ of $\sfW(\frL)$ which span $(\frL[2])^*$ and which are mapped to the curvature form components by adding particular monomials in the generators of $\sfW(\frL)$ given in terms of certain adjustment structure constants. This changes the definition of curvature, which in turn modifies the action of gauge and higher gauge transformations. One can then explicitly compute the requirements for commutators of gauge transformations to close and higher gauge transformations to preserve the images of gauge transformations, which amount to the conditions on the adjustment structure constants, the \uline{adjustment conditions}. An $L_\infty$-algebra, together with adjustment structure constants is called an \uline{adjusted $L_\infty$-algebra}, and we have the following proposition.
		
		\begin{proposition}(\cite[Proposition 2.1]{Gagliardo:2025oio})\label{prop:adjustment_L_infty}
			Given an $L_\infty$-algebra $\frL$, an automorphism $\phi:\sfW(\frL)\rightarrow\sfW(\frL)$ on its Weil algebra that covers the identity on the Chevalley--Eilenberg algebra (i.e.~$\phi(\tte^A)=\tte^A$) turns $\frL$ into an adjusted $L_\infty$-algebra if and only if
			\begin{equation}\label{eq:diff_form}
				\sfd_\sfW(\phi(\hat\tte^A))\ =\ \sum_{i\geq1}\frac1{i!}\Big(f_{B_1\cdots B_i}{}^A\phi(\hat\tte^{B_1})\cdots\phi(\hat\tte^{B_i})+g_{B_0B_1\cdots B_i}{}^A\tte^{B_0}\phi(\hat\tte^{B_1})\cdots\phi(\hat\tte^{B_i})\Big)
			\end{equation}
			for all shifted generators $\hat\tte^A=\sigma(\tte^A)$, where $f_{B_1\cdots B_i}{}^A$ and $g_{B_0B_1\ldots B_i}{}^A$ are structure constants with $g_{B_0B_1\cdots B_i}{}^A$ identically vanishing unless $|\tte^{B_0}|=1$.
		\end{proposition}
		
		The proof of this proposition is summarised in \cref{app:proof}. The resulting connections with modified curvatures and (higher) gauge transformations are called \uline{adjusted connections}.
		
		\paragraph{Example: 2-term $L_\infty$-algebras.}
		Consider a 2-term $L_\infty$-algebra $\frL=\frL_{-1}\oplus \frL_0$ with corresponding Weil algebra~\eqref{eq:Weil_2_term}. Besides a trivial rescaling of the generators $\hat\tte^\alpha$, the most general deformation is given by
		\begin{subequations}
			\begin{equation}
				\tilde{\hat\tte}^a\ \coloneqq\ \phi(\hat\tte^a)\ \coloneqq\ \hat\tte^a+\kappa_{\alpha\beta}{}^a\tte^\alpha\hat\tte^\beta~,
			\end{equation}
			where the structure constants $\kappa_{\alpha\beta}{}^a$ define a bilinear map
			\begin{equation}
				\kappa\,:\,\frL_0\times\frL_0\ \rightarrow\ \frL_{-1}~,
			\end{equation}
		\end{subequations}
		called the \uline{adjustment}. Whilst the 2-form curvature $F$ and the gauge transformation of the 1-form gauge potential $A$ remain unchanged, see \cref{ssec:higher_CS}, the curvature 3-form $H$ and the gauge transformation for 2-form gauge potential $B$ become~\cite{Saemann:2019dsl,Rist:2022hci}
		\begin{subequations}\label{eq:kinematicChangesDueToAdjustment}
			\begin{equation}
				\begin{aligned}
					H\ &=\ \rmd B+\mu_2(A,B)-\tfrac1{3!}\mu_3(A,A,A)-\kappa(A,F)~,
					\\
					\delta B\ &=\ \rmd\lambda+\mu_2(A,\lambda)+\mu_2(B,\alpha)+\tfrac12\mu_3(A,A,\alpha)+\kappa(\alpha,F)~.
				\end{aligned}
			\end{equation}
			Furthermore, higher gauge transformations are not modified; this is easily understood as any modification has to be proportional to $F$ and hence contain a 2-form. The Bianchi identities become
			\begin{equation}
				\begin{aligned}
					\rmd F+\mu_2(A,F)-\mu_1(\kappa(A,F))-\mu_1(H)\ &=\ 0~,
					\\
					\rmd H+\mu_2(A,H)-\kappa(A,\mu_1(H))+\kappa(F,F)\ &=\ 0~.
				\end{aligned}
			\end{equation}
		\end{subequations}
		
		The adjustment conditions arising from demanding that the commutator of two gauge transformations closes and that higher gauge transformations preserve the image of gauge transformations read as
		\begin{equation}\label{eq:adjustment_conditions_Lie_2}
			\begin{aligned}
				\kappa(\mu_1(Y),X)\ &=\ -\mu_2(X,Y)~,
				\\
				\kappa(\mu_2(X_1,X_2),X_3)\ &=\ \kappa(X_1,\mu_2(X_2,X_3)-\mu_1(\kappa(X_2,X_3)))
				\\
				&\hspace{1cm}-\kappa(X_2,\mu_2(X_1,X_3)-\mu_1(\kappa(X_1,X_3)))
				\\
				&\hspace{1cm}+\mu_2(X_1,\kappa(X_2,X_3))-\mu_2(X_2,\kappa(X_1,X_3))+\mu_3(X_1,X_2,X_3)
			\end{aligned}
		\end{equation}
		for all $Y\in\frL_{-1}$ and $X,X_{1,2,3}\in\frL_0$. This is a truncation of the adjustment for 3-term $L_\infty$-algebras discussed in~\cite{Gagliardo:2025oio}; see there for details (as well as~\cite{Saemann:2019dsl,Rist:2022hci} for the strict case).
		
		\paragraph{Example: adjusted string Lie 2-algebra.}
		For example, in the case of the string Lie 2-algebra for a metric Lie algebra $(\frg,[-,-],\inner{-}{-})$ discussed in \cref{ssec:cyclic_L_infty}, the quadratic form $\inner{-}{-}$ provides an adjustment
		\begin{equation}\label{eq:adjustment_skeletal_string}
			\kappa(X_1,X_2)\ \coloneqq\ \inner{X_1}{X_2}
		\end{equation}
		for all $X_{1,2}\in\frL_0=\frg$. In this case, the adjustment changes the 3-form curvature according to
		\begin{subequations}\label{eq:string_field_strength}
			\begin{equation}
				H\ =\ \rmd B-\tfrac1{3!}\inner{A}{[A,A]}
				\quad\rightarrow\quad
				H\ =\ \rmd B\underbrace{-\tfrac1{3!}\inner{A}{[A,A]}+\kappa(A,F)}_{=\,2{\rm cs}(A)}~,
			\end{equation}
			where
			\begin{equation}
				{\rm cs}(A)\ \coloneqq\ \tfrac12\inner{A}{\rmd A}+\tfrac1{3!}\inner{A}{[A,A]}
			\end{equation}
		\end{subequations}
		is the usual Chern--Simons 3-form. This expression for the 3-form $H$ is the one also appearing in heterotic supergravity~\cite{Bergshoeff:1981um,Chapline:1982ww}. The Bianchi identity becomes the well-known \uline{Green--Schwarz anomaly cancellation condition},
		\begin{equation}
			\rmd H\ =\ \inner{F}{F}~.
		\end{equation}
		
		Further physically relevant examples can be found in~\cite{Borsten:2021ljb}, where the adjustments appearing in the tensor hierarchy of gauged supergravity theories are discussed.
		
		\subsection{No adjustments for cyclic, minimal, non-Abelian \texorpdfstring{$n$}{n}-term \texorpdfstring{$L_\infty$}{L infinity}-algebras}\label{ssec:no_adjusted_cyclic_L-infty}
		
		The construction of higher Chern--Simons theories now seems to be straightforward, following a simple recipe: Consider a cyclic $n$-term $L_\infty$-algebra, add an adjustment, tensor it with the de~Rham complex on an $(n+2)$-dimensional manifold, and write down the homotopy Maurer--Cartan action~\eqref{eq:higherCSAction}, modifying the action appropriately to account for the deformed curvatures. 
		
		Unfortunately, there is a problem with this construction. By \cref{thm:equivalence}, we can reduce the higher gauge algebra of a higher Chern--Simons theory to its minimal model, obtaining an equivalent higher Chern--Simons theory. As we will show in the following, there are no cyclic minimal models with adjustment, except for Abelian ones with $\mu_{i>1}=0$.
		
		\paragraph{Example: Skeletal 2-term $L_\infty$-algebras.}
		Before proving the general statement, we start with the more transparent example of a skeletal (i.e.~$\mu_1=0$) cyclic 2-term $L_\infty$-algebra $\frL=\frL_{-1}\oplus\frL_0$. The first adjustment condition in~\eqref{eq:adjustment_conditions_Lie_2} reduces to $\mu_2(X,Y)=0$ for all $Y\in\frL_{-1}$ and $X\in\frL_0$. Cyclicity amounts to
		\begin{equation}
			\inner{X_1}{\mu_2(X_2,Y)}\ =\ \inner{\mu_2(X_1,X_2)}{Y}
		\end{equation}
		for all $Y\in\frL_{-1}$ and $X_{1,2}\in\frL_0$, and non-degeneracy of the inner product then implies that $\mu_2$ is trivial. The second adjustment condition in~\eqref{eq:adjustment_conditions_Lie_2} then reduces to
		\begin{equation}
			\mu_3(X_1,X_2,X_3)\ =\ 0~,
		\end{equation}
		so that all higher products $\mu_i$ on $\frL$ necessarily vanish, and $\frL$ is Abelian.
		
		\paragraph{General case.} We obtain the same result when we consider the adjustment conditions of a 3-term $L_\infty$-algebra as derived in~\cite{Gagliardo:2025oio}: cyclic adjusted skeletal 3-term $L_\infty$-algebras are Abelian. Generally, we have the following result.
		
		\begin{theorem}\label{thm:only_abelian_cyclic_adjustments}
			Any cyclic adjusted skeletal $n$-term $L_\infty$-algebra with $n>1$ is Abelian, that is, all the higher products $\mu_{i>1}$ are trivial.
		\end{theorem}
		
		\begin{proof}
			We start from \cref{prop:adjustment_L_infty}. Up to an irrelevant isomorphism on the subalgebra generated by the $\hat\tte^A$, we can write
			\begin{equation}
				\tilde{\hat\tte}^A\ \coloneqq\ \phi(\hat\tte^A)\ \coloneqq\ \hat\tte^A+\kappa_{IJ}{}^A\tte^I{\hat\tte}^J~,
			\end{equation}
			where $I$ and $J$ are multi-indices of positive length, e.g.~$I=(A_1\ldots A_i)$ with $i>0$, and we sum over multi-indices of arbitrary length.
			
			Using this notation and~\eqref{eq:WeilDifferential}, we have
			\begin{equation}
				\sfd_\sfW\hat\tte^A\ =\ -\sigma(\sfd_\sfCE\tte^A)\ =\ -\sigma(f_I{}^A\tte^I)\ =\ f_{B\underline{I}}{}^A\hat\tte^B\tte^{\underline{I}}\ =\ f_{B\underline{I}}{}^A\tilde{\hat\tte}^B\tte^{\underline{I}}-f_{B\underline{I}}{}^A\kappa_{JK}{}^B\tte^J\tilde{\hat\tte}^K\tte^{\underline{I}}~,
			\end{equation}
			where an underline as in ${\underline{I}}$ denotes a multi-index that can be of length zero. We further compute
			\begin{equation}\label{eq:Weil_Bianchi}
				\begin{aligned}
					\sfd_\sfW\tilde{\hat\tte}^A\ &=\ \sfd_\sfW\hat\tte^A+\kappa_{IJ}{}^A\sfd_\sfW(\tte^I\tilde{\hat\tte}^J)
					\\
					&=\ \sfd_\sfW\hat\tte^A+\kappa_{B\underline{I}J}{}^Af^B_K\tte^K\tte^{\underline{I}}\tilde{\hat\tte}^J+(-1)^{|I|}\kappa_{IB\underline{J}}{}^A\tte^I(\sfd_\sfW\tilde{\hat\tte}^B)\tilde{\hat\tte}^{\underline{J}}
					\\
					&=\ f_{B\underline{I}}{}^A\tilde{\hat\tte}^B\tte^{\underline{I}}-f_{B\underline{I}}{}^A\kappa_{JK}{}^B\tte^J\tilde{\hat\tte}^K\tte^{\underline{I}}+\kappa_{B\underline{I}J}{}^Af^B_K\tte^K\tte^{\underline{I}}\tilde{\hat\tte}^J+(-1)^{|I|}\kappa_{IB\underline{J}}{}^A\tte^I(\sfd_\sfW\tilde{\hat\tte}^B)\tilde{\hat\tte}^{\underline{J}}~,
				\end{aligned}
			\end{equation}
			where $|I|$ denotes the degree of the product of generators with index contained in the multi-index $I$.
			
			We note that~\eqref{eq:Weil_Bianchi} is of the form
			\begin{equation}\label{eq:from_Bianchi}
				\sfd_\sfW\tilde{\hat\tte}^A\ =\ \sum_{i\geq1}\frac1{i!}\Big(f_{B_0B_1\cdots B_i}{}^A\tilde{\hat\tte}^{B_0}\tte^{B_1}\cdots\tte^{B_i}+K^A_i\Big)\,,
			\end{equation}
			where $K^A_i$ is an expression involving the $\kappa_{IJ}{}^A$ as well as structure constants $f_{C_1\cdots C_j}{}^B$ with $j<i$. Consider the case $i=2$. Since $\frL$ is skeletal, $f_B{}^A=0$ and hence $K^A_2=0$. Condition~\eqref{eq:diff_form} then implies that $f_{B_0B_1}{}^A=0$ for $|\tte^{B_1}|>1$, and hence $\mu_2:\frL\times\frL_i\rightarrow\frL$ vanishes for $i<0$. Together with cyclicity, it follows that $\mu_2$ vanishes altogether. Considering~\eqref{eq:from_Bianchi} for $i>3$ then shows similarly that all the higher products vanish, and, thus, $\frL$ is Abelian.
		\end{proof}
		
		\section{Half-adjusted higher Chern--Simons theories}
		
		In this section, we present a definition of higher Chern--Simons theories that circumvents the above problems, but satisfies the three desired properties of such theories outlined in the introduction. The starting point is an interesting adjusted $L_\infty$-algebra, for example the string Lie 2-algebra, which we complete to a cyclic $L_\infty$-algebra by the shifted cotangent construction~\eqref{eq:shifted_cotangent_construction}. Not all higher gauge transformations will close, as predicted by \ref{thm:only_abelian_cyclic_adjustments}. Surprisingly, all ordinary gauge transformations do, in fact, close, and the non-closing transformations are merely the higher gauge transformations in the cotangent direction. These can be put to zero manually, and we can develop the differential cohomology describing higher principal bundles that carry such connections.
		
		\subsection{Four-dimensional case: Infinitesimal considerations}\label{ssec:four_dimensional_case_infinitesimal_considerations}
		
		First, let us focus on the four-dimensional case, as this is most instructive. Here, the gauge $L_\infty$-algebra is a general 2-term $L_\infty$-algebra.
		
		\paragraph{Half-adjusted $L_\infty$-algebra.}
		Consider a general adjusted 2-term $L_\infty$-algebra\footnote{The most instructive example to have in mind is the string Lie 2-algebra introduced in \cref{ssec:cyclic_L_infty} with $\kappa$ the Cartan--Killing form.}
		\begin{equation}
			\frL\ =\ \frL_{-1}\oplus\frL_0
			\eand
			\kappa\,:\,\frL_0\times\frL_0\ \rightarrow\ \frL_{-1}~.
		\end{equation}
		Let $\kappa_+$ (respectively, $\kappa_-$) be the symmetric (respectively, anti-symmetric) part of $\kappa$. In the following, we shall find useful the maps
		\begin{equation}\label{eq:infinitesimalUsefulMaps}
			\begin{aligned}
				\nu_2(X,Y)\ &\coloneqq\ \mu_2(X,Y)-\kappa(X,\mu_1(Y))\ \in\ \frL_{-1}~,
				\\
				\nu_{2\pm}(X,Y)\ &\coloneqq\ \mu_2(X,Y)-\kappa_\pm(X,\mu_1(Y))\ \in\ \frL_{-1}~,
				\\
				\tilde\nu_2(X_1,X_2)\ &\coloneqq\ \mu_2(X_1,X_2)-\mu_1(\kappa(X_1,X_2))\ \in\ \frL_0~,
				\\
				\tilde\nu_{2\pm}(X_1,X_2)\ &\coloneqq\ \mu_2(X_1,X_2)-\mu_1(\kappa_\pm(X_1,X_2))\ \in\ \frL_0
			\end{aligned}
		\end{equation}
		for all $Y\in\frL_{-1}$ and $X,X_{1,2}\in\frL_0$. The various identities these maps satisfy are collated in \cref{app:half-adjusted-gauge}.
		
		Next, we consider the associated cotangent $L_\infty$-algebra $T^*[1]\frL$ as explained in \cref{ssec:cyclic_L_infty}. We have
		\begin{equation}
			(T^*[1]\frL)_k\ =\
			\begin{cases}
				\frL_{-1}\oplus(\frL_0)^* &\efor k\ =\ -1~,
				\\
				\frL_0\oplus(\frL_{-1})^* &\efor k\ =\ 0~.
			\end{cases}
		\end{equation}
		Furthermore, we have the non-vanishing dual maps 
		\begin{equation}\label{eq:dualMapsHalfAdjusted2Term}
			\begin{aligned}
				\mu^*_1(X^*)(Y)\ &=\ X^*(\mu_1(Y))~,
				\\
				\mu^*_2(X_1,X^*)(X_2)\ &=\ -X^*(\mu_2(X_1,X_2))~,
				\\
				\mu^*_2(X,Y^*_1)(Y_2)\ &=\ Y^*_1(\mu_2(X,Y_2))~,
				\\
				\mu^*_3(X_1,X_2,Y^*)(X_3)\ &=\ -Y^*(\mu_3(X_1,X_2,X_3))~,
				\\
				\nu^*_2(X,Y^*_1)(Y_2)\ &=\ - Y^*_1(\nu_2(X,Y_2))~,
				\\
				\tilde\nu^*_2(X_1,X^*_2)(X_3)\ &=\ -X^*_2(\tilde\nu_2(X_1,X_3))~,
				\\
				\kappa^*(X_1,Y^*)(X_2)\ &=\ -Y^*(\kappa(X_1,X_2))
			\end{aligned}
		\end{equation}
		for all $Y,Y_{1,2}\in\frL_{-1}$, $X,X_{1,2,3}\in\frL_0$, $Y^*\in(\frL_{-1})^*$, and $X^*\in(\frL_0)^*$.
		
		We call $\hat\frL$ a \uline{half-adjusted $L_\infty$-algebra} because whilst $\hat \frL$ is generically not an adjusted $L_\infty$-algebra, the sub-$L_\infty$-algebra $\frL\subseteq\hat\frL$ is.
		
		\paragraph{Field content and action.}
		In the following, let $M$ be a four-dimensional compact oriented manifold without boundary. The field content consists of $1$- and $2$-form components\footnote{Here, we slightly abuse notation by denoting forms taking values in the cotangent spaces by an asterisk, and there is no relation, for example, between $A$ and $A^*$.}
		\begin{equation}
			\begin{aligned}
				\hat A\ &=\ (A,A^*)\ \in\ \Omega^1(M)\otimes\frL_0\oplus\Omega^1(M)\otimes(\frL_{-1})^*~,
				\\
				\hat B\ &=\ (B,B^*)\ \in\ \Omega^2(M)\otimes\frL_{-1}\oplus\Omega^2(M)\otimes\frL_0^*~.
			\end{aligned}
		\end{equation}
		Furthermore, we have the curvatures
		\begin{subequations}
			\begin{equation}
				\begin{aligned}
					F\ &\coloneqq\ \rmd A+\tfrac12\mu_2(A,A)+\mu_1(B)~,
					\\
					H\ &\coloneqq\ \rmd B+\mu_2(A,B)-\tfrac1{3!}\mu_3(A,A,A)-\kappa(A,F)
				\end{aligned}
			\end{equation}
			and the Bianchi identities
			\begin{equation}
				\begin{aligned}
					\rmd F+\mu_2(A,F)-\mu_1(\kappa(A,F))-\mu_1(H)\ &=\ 0~,
					\\
					\rmd H+\mu_2(A,H)+\kappa(F,F)-\kappa(A,\mu_1(H))\ &=\ 0~,
				\end{aligned}
			\end{equation}
		\end{subequations}
		which are the usual expressions for an $\frL$-valued connection form $(A,B)$, see~\eqref{eq:kinematicChangesDueToAdjustment}.
		
		We combine these fields into the action
		\begin{equation}\label{eq:action_CS_4d}
			S\ \coloneqq\ \int_M\big\{A^*(H)+B^*(F)\big\}\,,
		\end{equation}
		where wedge products are implied.
		
		\paragraph{Additional curvature forms and equations of motion.}
		In order to identify the additional curvature forms $F^*$ and $H^*$ taking values in the cotangent directions of $T^*[1]\frL$, we use the evident higher generalisation of the relation that identifies the differential of the Chern--Simons 3-form with the 4-form which integrates to the first Pontryagin class. That is, we regard $M$ as (part of) the boundary of a five-dimensional manifold $N$, extend all fields to $N$, and demand that\footnote{We shall not make a notational distinction for fields on $M$ and their extensions to $N$.}
		\begin{equation}\label{eq:cond_higher_pontryagin}
			\rmd\big(A^*(H)+B^*(F)\big)\ \overset{!}{=}\ F^*(H)+H^*(F)
		\end{equation}
		for the extended fields. We compute
		\begin{equation}\label{eq:compute_cotangent_curvatures}
			\begin{aligned}
				\rmd\big(A^*(H)+B^*(F)\big)\ &=\ (\rmd A^*)(H)-A^*(\rmd H)+(\rmd B^*)(F)+B^*(\rmd F)
				\\
				&=\ (\rmd A^*)(H)+A^*\big(\mu_2(A,H)-\kappa(A,\mu_1(H))\big)+B^*(\mu_1(H))
				\\
				&\kern1cm+(\rmd B^*)(F)-B^*\big(\mu_2(A,F)-\mu_1(\kappa(A,F))\big)+A^*(\kappa_+(F,F))~,
			\end{aligned}
		\end{equation}
		where, as defined above, $\kappa_+$ denotes the symmetric part of $\kappa$.
		
		From the computation~\eqref{eq:compute_cotangent_curvatures}, it now follows that the additional curvature forms are uniquely fixed to be 
		\begin{subequations}
			\begin{equation}
				\begin{aligned}
					F^*\ &=\ \rmd A^*+\nu^*_2(A,A^*)+\mu^*_1(B^*)~,
					\\
					H^*\ &=\ \rmd B^*+\tilde\nu^*_2(A,B^*)-\kappa^*_+(F,A^*)
				\end{aligned}
			\end{equation}
			with Bianchi identities
			\begin{equation}
				\begin{aligned}
					\rmd F^*+\nu^*_2(A,F^*)-\nu^*_{2-}(F,A^*)-\mu^*_1(H^*)\ &=\ 0~,
					\\
					\rmd H^*+\tilde\nu^*_2(A,H^*)-\tilde\nu^*_{2-}(F,B^*)+\kappa^*_+(\mu_1(H),A^*)+\kappa^*_+(F,F^*)\ &=\ 0~,
				\end{aligned}
			\end{equation}
		\end{subequations}
		where we have used the maps introduced in~\eqref{eq:dualMapsHalfAdjusted2Term}.
		
		Finally, upon varying the action~\eqref{eq:action_CS_4d}, we find the equations of motion
		\begin{equation}
			F\ =\ 0~,
			\quad
			F^*\ =\ 0~,
			\quad
			H\ =\ 0~,
			\eand
			H^*\ =\ 0~.
		\end{equation}
		
		We also have the following result.
		\begin{proposition}\label{prop:4d_gauge_invariance}
			Condition~\eqref{eq:cond_higher_pontryagin} ensures that the action~\eqref{eq:action_CS_4d} is invariant under infinitesimal gauge transformations.
		\end{proposition}
		\begin{proof}
			This directly follows from the interpretation of gauge transformations connected to the identity as partially flat homotopies as given in~\cite{Sati:2008eg}. Concretely, the infinitesimal gauge transformation of the gauge potentials $a\in\{A,A^*,B,B^*\}$ on $N$ can be computed by formally extending $a$ to $N\times[0,1]$ with the new homotopy direction parametrised by $t\in[0,1]$ and with the infinitesimal gauge transformation $\delta_ca$ of $a$ given by $\delta_ca=\parder{t}\big|_{t=0}a$ under the assumption $\parder{t}\intprod f=0$ for the curvatures $f\in\{F,F^*,H,H^*\}$ and with the gauge parameters given $c\coloneqq\parder{t}\intprod a\big|_{t=0}$. Because $\parder{t}\intprod f$ contains all terms $\parder{t}a$, the condition $\parder{t}\intprod f=0$ allows us to determine all $\delta_ca$.
			
			From this perspective, it is evident that the differential $\rmd \caL$ of the Lagrangian $\caL=A^*(H)+B^*(F)$ is gauge invariant on $N\times[0,1]$. Considering the integral of this form over $N\times [0,1]$, we can use Stokes' theorem to see that the difference of the two boundary contributions, which is a difference of terms of the form~\eqref{eq:action_CS_4d}, is also gauge invariant. Since we can choose the gauge parameters such that they become trivial on one of the boundaries, it follows that each boundary is gauge invariant by itself.
		\end{proof}

		\paragraph{Gauge transformations.}
		As explained above, the gauge transformations are fully induced by the choice of curvatures. With the choices made above, infinitesimal gauge transformations are parametrised by
		\begin{subequations}\label{eq:inf_gauge_trafos}
			\begin{equation}
				\begin{aligned}
					\hat\alpha\ =\ (\alpha,\alpha^*)\ \in\ \Omega^0(M)\otimes\frL_0\oplus\Omega^0(M)\otimes(\frL_{-1})^*~,
					\\
					\hat\lambda\ =\ (\lambda,\lambda^*)\ \in\ \Omega^1(M)\otimes\frL_{-1}\oplus\Omega^1(M)\otimes(\frL_0)^*~,
				\end{aligned}
			\end{equation}
			and act according to
			\begin{equation}
				\begin{aligned}
					\delta_{(\alpha,\alpha^*,\lambda,\lambda^*)}A\ &=\ \rmd\alpha+\mu_2(A,\alpha)-\mu_1(\lambda)~,
					\\
					\delta_{(\alpha,\alpha^*,\lambda,\lambda^*)}A^*\ &=\ \rmd\alpha^*+\nu^*_2(A,\alpha^*)-\nu^*_2(\alpha,A^*)-\mu^*_1(\lambda^*)~,
					\\
					\delta_{(\alpha,\alpha^*,\lambda,\lambda^*)}B\ &=\ \rmd\lambda+\mu_2(A,\lambda)-\mu_2(\alpha,B)+\tfrac12\mu_3(A,A,\alpha)+\kappa(\alpha,F)~,
					\\
					\delta_{(\alpha,\alpha^*,\lambda,\lambda^*)}B^*\ &=\ \rmd\lambda^*+\tilde\nu^*_2(A,\lambda^*)-\tilde\nu^*_2(\alpha, B^*)+\kappa^*_+(F, \alpha^*)~.
				\end{aligned}
			\end{equation}
			The curvature forms transform correspondingly as
			\begin{equation}
				\begin{aligned}
					\delta_{(\alpha,\alpha^*,\lambda,\lambda^*)}F\ &=\ -\tilde\nu_2(\alpha,F)~,
					\\
					\delta_{(\alpha,\alpha^*,\lambda,\lambda^*)}F^*\ &=\ -\nu^*_2(\alpha,F^*)+\nu^*_{2-}(F,\alpha^*)~,
					\\
					\delta_{(\alpha,\alpha^*,\lambda,\lambda^*)}H\ &=\ -\nu_2(\alpha,H)~,
					\\
					\delta_{(\alpha,\alpha^*,\lambda,\lambda^*)}H^*\ &=\ -\tilde\nu^*_2(\alpha,H^*)+\nu^*_{2-}(F,\lambda^*)+\kappa^*_+(\mu_1(H),\alpha^*)~.
				\end{aligned}
			\end{equation}
		\end{subequations}
		
		We can now compute the commutator of two gauge transformations acting on the gauge potentials $a\in\{A,A^*,B,B^*\}$. We obtain
		\begin{subequations}\label{eq:closure_of_gauge}
			\begin{equation}
				\big[\delta_{(\alpha_1,\alpha^*_1,\lambda_1,\lambda^*_1)},\delta_{(\alpha_2,\alpha^*_2,\lambda_2,\lambda^*_2)}\big]\ =\ \delta_{(\alpha_3,\alpha^*_3,\lambda_3,\lambda^*_3)}
			\end{equation}
			with 
			\begin{equation}
				\begin{aligned}
					\alpha_3\ &\coloneqq\ \mu_2(\alpha_2,\alpha_1)~,
					\\
					\alpha^*_3\ &\coloneqq\ \nu^*_2(\alpha_2,\alpha^*_1)-\nu^*_2(\alpha_1,\alpha^*_2)~,
					\\
					\lambda_3\ &\coloneqq\ \mu_2(\alpha_2,\lambda_1)-\mu_2(\alpha_1,\lambda_2)-\mu_3(A,\alpha_2,\alpha_1)~,
					\\
					\lambda^*_3\ &\coloneqq\ \tilde\nu^*_2(\alpha_2,\lambda^*_1)-\tilde\nu^*_2(\alpha_1,\lambda^*_2)~.
				\end{aligned}
			\end{equation}
		\end{subequations}
		The details of the computation are found in \cref{app:half-adjusted-gauge}.
		
		\paragraph{Higher gauge transformations.}
		Higher gauge transformations are derived either via higher homotopies, as described e.g.~in~\cite{Jurco:2018sby} or via the adjusted BRST complex described in \cref{ssec:no_adjusted_cyclic_L-infty}. These are parametrised by $\hat\theta=(\theta,\theta^*)\in\Omega^0(M)\otimes\frL_{-1}\oplus\Omega^0(M)\otimes(\frL_0)^*$ and act as
		\begin{equation}\label{eq:inf_higher_gauge_trafos}
			\begin{aligned}
				\delta_{(\theta,\theta^*)}\alpha\ &=\ \mu_1(\theta)~,
				\\
				\delta_{(\theta,\theta^*)}\alpha^*\ &=\ \mu^*_1(\theta^*)~,
				\\
				\delta_{(\theta,\theta^*)}\lambda\ &=\ \rmd\theta+\mu_2(A,\theta)~,
				\\
				\delta_{(\theta,\theta^*)}\lambda^*\ &=\ \rmd\theta^*+\tilde\nu^*_2(A,\theta^*)~.
			\end{aligned}
		\end{equation}
		From the discussion in \cref{ssec:problems}, we now expect that two gauge transformations that are higher gauge equivalent, i.e.~linked by a higher gauge transformation, will have different images. Defining
		\begin{subequations}
			\begin{equation}
				\delta\ \coloneqq\ \delta_{(\alpha,\alpha^*,\lambda,\lambda^*)}-\delta_{(\alpha,\alpha^*,\lambda,\lambda^*)+\delta_{(\theta,\theta^*)}(\alpha,\alpha^*,\lambda,\lambda^*)}~,
			\end{equation}
			we indeed obtain
			\begin{equation}
				\begin{aligned}
					\delta A\ &=\ 0~,
					\\
					\delta A^*\ &=\ 0~,
					\\
					\delta B\ &=\ \mu_2(F,\theta)+\kappa(\mu_1(\theta),F)\ =\ 0~,
					\\
					\delta B^*\ &=\ \mu^*_2(F,\theta^*)-\kappa^*_-(F,\mu_1(\theta^*))\big)\ \neq\ 0~.
				\end{aligned}
			\end{equation}
		\end{subequations}
		Therefore, we have to restrict higher gauge transformations to those that are parametrised by $\hat\theta=(\theta,0)$.\footnote{Alternatively, we could implement corresponding trivial symmetries as explained in \ref{ssec:trivial_symmetries}, which has the same effect.}
		
		\paragraph{Example.} The archetypal example to consider in the infinitesimal case is the half-adjustment constructed from the adjusted string Lie 2-algebra~\eqref{eq:adjustment_skeletal_string}. Note that here, 
		\begin{equation}
			\mu_1(Y)\ =\ 0\eand \mu_2(X,Y)\ =\ 0
		\end{equation} 
		which implies that 
		\begin{equation}
			\nu_2(X,Y)\ =\ \nu_{2\pm}(X,Y)\ =\ 0\eand \tilde\nu_2(X_1,X_2)\ =\ \tilde\nu_{2\pm}(X_1,X_2)\ =\ \mu_2(X_1,X_2)
		\end{equation} 
		for all $X,X_{1,2}\in \frL_0=\frg$ and $Y\in \frL_{-1}=\IR$. As a result, the action~\eqref{eq:action_CS_4d} becomes
		\begin{equation}
			S\ =\ \int_{M} \left\{A^*(\rmd B+\inner{A}{\rmd A}+\tfrac13\inner{A}{[A,A]})+B^*(\rmd A+\tfrac12[A,A])\right\}
		\end{equation}
		for $M$ some four-dimensional manifold. Let us stress that although we only half-adjusted the gauge algebra, the gauge transformations still all close. In line with our theorem, however, higher gauge transformations still link gauge transformations with different images and therefore have to be excluded.
		
		\subsection{Four-dimensional case: Finite considerations}
		
		Let us also consider the global description of half-adjusted higher Chern--Simons theory, by developing the explicit form of finite gauge transformations for the kinematic data of half-adjusted four-dimensional Chern--Simons theory. We note that, as usual for most discussions on Chern--Simons theory, our action functional is only suitable for kinematic data given by connections on topologically trivial principal 2-bundles, i.e.~principal 2-bundles with trivial higher transition functions. Nevertheless, and for future use, we also compute the differential cohomology describing general higher principal 2-bundles with half-adjusted connections.
		
		As usual, we restrict ourselves to a \uline{strict} gauge 2-group to keep the computations manageable. This is a mild restriction, as general theorems allows us to turn any 2-group into a weakly equivalent strict gauge 2-group~\cite{Baez:0307200}.
		
		\paragraph{Cotangent crossed module of Lie algebras.}
		The starting point of our discussion is the cotangent 2-term $L_\infty$-algebra $T^*[1]\frL$ introduced in \cref{ssec:cyclic_L_infty}. We shall now assume that the underlying 2-term $L_\infty$-algebra $\frL$ is strict, that is, $\mu_3=0$. Consequently, $T^*[1]\frL$ is also strict, that is, $\hat\mu_3=0$. In \cref{ssec:cyclic_L_infty}, we have also seen that any strict 2-term $L_\infty$-algebra is equivalent to a crossed module of Lie algebras, and we will use this perspective in the following, as it is more convenient.
		
		Firstly, we denote the crossed module of Lie algebras equivalent to $\frL$ by $(\frh\overset{\sft}{\rightarrow}\frg,\acton)$. In particular, $\frh\coloneqq\frL_{-1}$, $\frg\coloneqq\frL_0$, and $\sft\coloneqq\mu_1$. Furthermore, $[X_1,X_2]\coloneqq\mu_2(X_1,X_2)$, $[Y_1,Y_2]\coloneqq\mu_2(\mu_1(Y_1),Y_2)$, and $X\acton Y\coloneqq\mu_2(X,Y)$. Next, we set
		\begin{subequations}\label{eq:cotangentCrossedModule}
			\begin{equation}\label{eq:cotangentCrossedModule:a}
				\hat\frh\ \coloneqq\ (T^*[1]\frL)_{-1}\ =\ \frh\oplus\frg^*
				\eand
				\hat\frg\ \coloneqq\ (T^*[1]\frL)_0\ =\ \frg\oplus\frh^*
			\end{equation}
			together with
			\begin{equation}\label{eq:cotangentCrossedModule:b}
				\begin{aligned}
					\hat\sft\,:\,\hat\frh\ &\rightarrow\ \hat\frg~,
					\\
					(Y,X^*)\ &\mapsto\ \big(\sft(Y),\sft^*(X^*)\big)
				\end{aligned}
			\end{equation}
			and
			\begin{equation}\label{eq:cotangentCrossedModule:c}
				(X_1,Y_1^*)\,\hat\acton\,(Y_2,X_2^*)\ \coloneqq\ \big(X_1\acton Y_2,[X_1,X^*_2]^*-[Y_2,Y^*_1]^*\big)
			\end{equation}
			as well as
			\begin{equation}\label{eq:cotangentCrossedModule:d}
				\big[(Y_1,X^*_1),(Y_2,X^*_2)\big]\ \coloneqq\ \big([Y_1,Y_2],[\sft(Y_1),X^*_2]^*-[\sft(Y_2),X^*_1]^*\big)
			\end{equation}
			and
			\begin{equation}\label{eq:cotangentCrossedModule:e}
				\big[(X_1,Y^*_1),(X_2,Y^*_2)\big]\ \coloneqq\ \big([X_1,X_2],[X_1,Y^*_2]^*-[X_2,Y^*_1]^*\big)
			\end{equation}
		\end{subequations}
		for all $(Y_{1,2},X_{1,2}^*)\in\hat\frh$ and $(X_{1,2},Y_{1,2}^*)\in\hat\frg$. Here, $[-,-]^*$ denotes the evident dual of the adjoint action. It is then straightforward to check that all the axioms of a crossed module of Lie algebras are satisfied, that is, the brackets~\eqref{eq:cotangentCrossedModule:d} and~\eqref{eq:cotangentCrossedModule:e} are Lie brackets on $\hat\frh$ and $\hat\frg$, the map~\eqref{eq:cotangentCrossedModule:b} is a morphism of Lie algebras, and the map~\eqref{eq:cotangentCrossedModule:b} and the action~\eqref{eq:cotangentCrossedModule:c} satisfy the equivariance and Peiffer conditions. We shall refer to the crossed module of Lie algebras $(\hat\frh\overset{\hat\sft}{\rightarrow}\hat\frg,\hat\acton)$ as the \uline{cotangent crossed module of Lie algebras} associated with $(\frh\overset{\sft}{\rightarrow}\frg,\acton)$, and we shall denote it more succinctly by $T^*[1](\frh\rightarrow\frg)$.
		
		\paragraph{Cotangent crossed module of Lie groups.}
		Let $(\sfH\overset{\sft}{\rightarrow}\sfG,\acton)$ be a crossed module of Lie groups integrating $(\frh\overset{\sft}{\rightarrow}\frg,\acton)$. Then, there is a \uline{cotangent crossed module of Lie groups}, denoted by $T^*[1](\sfH\rightarrow\sfG)$, integrating $T^*[1](\frh\rightarrow\frg)$. Recall that $T\sfH$ (respectively, $T\sfG$) is isomorphic to $\sfH\times\frh$ (respectively, $\sfG\times\frg$) so that Lie groups associated with the Lie algebras $\hat\frh$ and $\hat\frg$ in~\eqref{eq:cotangentCrossedModule:a} are 
		\begin{subequations}
			\begin{equation}
				\hat\sfH\ \coloneqq\ \sfH\times\frg^*
				\eand
				\hat\sfG\ \coloneqq\ \sfG\times\frh^*~,
			\end{equation}
			respectively, with group products given by 
			\begin{equation}
				(h_1,X^*_1)(h_2,X^*_2)\ \coloneqq\ \big(h_1h_2,X_1^*+\sft(h)X_2^*(\sft(h))^{-1}\big)
			\end{equation}
			and
			\begin{equation}
				(g_1,Y^*_1)(g_2,Y^*_2)\ \coloneqq\ \big(g_1g_2,Y_1^*+g_1\acton Y_2^*\big)
			\end{equation}
			for all $Y_{1,2}^*\in\frh^*$, $h_{1,2}\in\sfH$, $X_{1,2}^*\in\frg^*$, and $g_{1,2}\in\sfG$. In addition,
			\begin{equation}
				\begin{aligned}
					\hat\sft\,:\,\hat\sfH\ &\rightarrow\ \hat\sfG~,
					\\
					(h,X^*)\ &\mapsto\ \big(\sft(h),\sft^*(X^*)\big)
				\end{aligned}
			\end{equation}
			and
			\begin{equation}
				(g,Y^*)\,\hat\acton\,(h,X^*)\ \coloneqq\ \big(g\acton h,gX^*g^{-1}-(g\acton h)Y^*(g\acton h)^{-1}\big)
			\end{equation}
			for all $(h,X^*)\in\hat\sfH$ and $(g,Y^*)\in\hat\sfG$. Here,
			\begin{equation}
				\begin{gathered}
					(gX^*g^{-1})(X_1)\ \coloneqq X^*(gX_1g^{-1})~,
					\quad
					(hY^*h^{-1})(X)\ \coloneqq\ Y^*(h(X\acton h^{-1}))~,
					\\
					(g\acton Y^*)(Y_1)\ \coloneqq\ Y^*(g\acton Y_1)
				\end{gathered}
			\end{equation}
		\end{subequations}
		for all $Y_1\in\frh$, $Y^*\in\frh^*$, $h\in\sfH$, $X,X_1\in\frg$, $X^*\in\frg^*$, and $g\in\sfG$. Using the identities
		\begin{equation}
			\begin{aligned}
				\sft^*(hY^*h^{-1})\ &=\ Y^*-\sft(h)\acton Y^*~,
				\\
				h(\sft^*(X^*)\acton h^{-1})\ &=\ X^*-\sft(h)X^*(\sft(h))^{-1}
			\end{aligned}
		\end{equation}
		for all $Y^*\in\frh^*$, $h\in\sfH$, and $X^*\in\sfG^*$, it is not too difficult to see that $T^*[1](\sfH\rightarrow\sfG)$ is indeed a crossed module of Lie groups. Similarly, one shows that this crossed module of Lie groups differentiates to $T^*[1](\frh\rightarrow\frg)$.
		
		\paragraph{Half adjustments.}
		So far, we have only discussed the cotangent 2-term $L_\infty$-algebra $T^*[1]\frL$, translated it to the crossed module language to obtain the crossed module of Lie algebras $T^*[1](\frh\rightarrow\frg)$, and integrated the latter to the crossed module of Lie groups $T^*[1](\sfH\rightarrow\sfG)$. We can now also adjust $\frL$ with an adjustment datum $\kappa$, or, equivalently, the associated crossed module of Lie algebras to obtain $(\frh\overset{\sft}{\rightarrow}\frg,\acton,\kappa)$, cf.~\cite{Rist:2022hci}. This integrates to an adjusted crossed module of Lie groups  $(\sfH\overset{\sft}{\rightarrow}\sfG,\acton,\kappa)$ in the sense of~\cite{Rist:2022hci}. That is, the map $\kappa$ satisfies
		\begin{subequations}
			\begin{align}
				\kappa(\sft(h),Y)\ &=\ h(Y\acton h^{-1})~,\label{eq:alternativeAdjustmentCondition_a}
				\\
				\kappa(g_2g_1,X)\ &=\ g_2\acton\kappa(g_1,X)+\kappa\big(g_2,g_1Xg^{-1}_1-\sft(\kappa(g_1,X))\big)\label{eq:alternativeAdjustmentCondition_b}
			\end{align}
		\end{subequations}
		for all $X\in \frg$, $Y\in \frh$ and $g_{1,2}\in \sfG$. Consequently, $T^*[1](\sfH\rightarrow\sfG)$ becomes half-adjusted following our discussion in \cref{ssec:four_dimensional_case_infinitesimal_considerations}. We call the resulting crossed modules of Lie groups the \uline{half-adjusted cotangent crossed module of Lie groups} and denote it by $T^*_\kappa[1](\sfH\rightarrow\sfG)$.
		
		\paragraph{Finite gauge transformations.}
		Consider $T^*_\kappa[1](\sfH\rightarrow\sfG)$. The next step in the derivation of the half-adjusted differential cocycles is the construction of the higher gauge groupoid. In particular, we need to integrate the infinitesimal gauge transformations~\eqref{eq:inf_gauge_trafos} as well as the infinitesimal higher gauge transformations~\eqref{eq:inf_higher_gauge_trafos}, which form a Lie 2-algebroid, into a 2-groupoid of finite gauge transformations, where we disregard 2-transformations in the cotangent direction. These computations are lengthy, but in principle straightforward. We therefore merely display the result.
		
		The maps~\eqref{eq:infinitesimalUsefulMaps} integrate to
		\begin{equation}\label{eq:nu_maps}
			\begin{aligned}
				\nu_2(g,Y)\ &=\ g\acton Y-\kappa(g,\sft(Y))\ \in\ \frh~,
				\\
				\tilde\nu_2(g,X)\ &=\ gXg^{-1}-\sft(\kappa(g,X))\ \in\ \frg
			\end{aligned}
		\end{equation}
		for all $Y\in\frh$, $X\in\frg$, and $g\in\sfG$. Likewise, the corresponding dual maps~\eqref{eq:dualMapsHalfAdjusted2Term} integrate to
		\begin{equation}
			\begin{aligned}
				\nu^*_2(g,Y^*)(Y_1)\ &=\ Y^*(\nu_2(g^{-1},Y_1))~,
				\\
				\tilde\nu^*_2(g,X^*)(X_1)\ &=\ X^*(\tilde\nu_2(g^{-1},X_1))
			\end{aligned}
		\end{equation}
		for all $Y_1\in\frh$, $Y^*\in\frh^*$, $X_1\in\frg$, $X^*\in\frg^*$, and $g\in\sfG$. These maps describe the deformation of the action induced by the adjustment datum, and they satisfy a number of identities listed in \cref{app:half-adjusted-gauge}.
		
		Finite gauge transformations now are parametrised by elements
		\begin{equation}
			\begin{aligned}
				\hat g\ &=\ (g,\Gamma^*)\ \in\ \scC^\infty(M,\sfG)\times\Omega^0(M)\otimes\frh^*~,
				\\
				\hat\Lambda\ &=\ (\Lambda,\Lambda^*)\ \in\ \Omega^1(M)\otimes\frh\oplus\Omega^1(M)\otimes\frg^*~,
			\end{aligned}
		\end{equation}
		and they transform gauge potentials according to
		\begin{equation}
			\begin{aligned}
				\tilde A\ &\coloneqq\ g^{-1}Ag+g^{-1}\rmd g-\sft(\Lambda)~,
				\\
				\tilde A^*\ &\coloneqq\ \nu^*_2\big(g^{-1},A^*+\rmd\Gamma^*+\nu^*_2(A,\Gamma^*)\big)-\sft^*(\Lambda^*)~,
				\\
				\tilde B\ &\coloneqq\ g^{-1}\acton B+\rmd\Lambda+\tilde{A}^0\acton\Lambda+\tfrac12[\Lambda,\Lambda]-\kappa\big(g^{-1},F\big)~,
				\\
				\tilde B^*\ &\coloneqq\ \tilde\nu^*_2\big(g^{-1},B^*+\kappa^*_+(F,\Gamma^*)\big)+\rmd\Lambda^*+\tilde\nu^*_2(\tilde A,\Lambda^*)~.
			\end{aligned}
		\end{equation}
		In turn, the curvature forms transform as
		\begin{equation}
			\begin{aligned}
				\tilde F\ &\coloneqq\ \tilde\nu_2(g^{-1},F)~,
				\\
				\tilde F^*\ &\coloneqq\ \nu^*_2(g^{-1},F^*)+\nu^*_2(F,\Gamma^*)-\sft^*\big((\tilde\nu_2g^{-1},\kappa^*_+(F,\Gamma^*)\big)-\kappa^*_+(F,\Gamma^*))~,
				\\
				\tilde H\ &\coloneqq\ \nu_2(g^{-1},H)~,
				\\
				\tilde H^*\ &\coloneqq\ \tilde\nu^*_2\big(g^{-1},H^*+\kappa^*_+(\sft(H),\Gamma^*)\big)+\kappa^*_+\big(\tilde\nu_2(g^{-1},F),\sft^*(\Lambda^*)\big)+\tilde\nu^*_2(g^{-1}F_0g,\Lambda^*)~.
			\end{aligned}
		\end{equation}
		
		The composition of two successive gauge transformations parametrised by $(g_{1,2},\Gamma^*_{1,2},\Lambda_{1,2},\Lambda^*_{1,2})$ is given by a third gauge transformation parametrised by $(g_3,\Gamma^*_3,\Lambda_3,\Lambda^*_3)$ with
		\begin{equation}
			\begin{aligned}
				g_3\ &\coloneqq\ g_1g_2~,
				\\
				\Gamma^*_3 &\coloneqq\ \Gamma^*_1+\nu^*_2(g_1,\Gamma^*_2)~,
				\\
				\Lambda_3\ &\coloneqq\ \Lambda_2+g_2^{-1}\acton\Lambda_1~,
				\\
				\Lambda^*_3\ &\coloneqq\ \Lambda^*_2+\tilde\nu^*_2(g_2^{-1},\Lambda^*_1)~.
			\end{aligned}
		\end{equation}
		This composition is associative, and the gauge transformation $(g,\Gamma^*,\Lambda,\Lambda^*)$ has the inverse gauge transformation parametrised by $(g_{\rm inv},\Gamma^*_{\rm inv},\Lambda_{\rm inv},\Lambda^*_{\rm inv})$ with
		\begin{equation}
			\begin{aligned}
				g_{\rm inv}\ &\coloneqq\ g^{-1}~,
				\\
				\Gamma^*_{\rm inv}\ &\coloneqq\ -\nu^*_2(g^{-1},\Gamma^*)~,
				\\
				\Lambda_{\rm inv}\ &\coloneqq\ -g\acton\Lambda~,
				\\
				\Lambda^*_{\rm inv}\ &\coloneqq\ -\tilde\nu^*_2(g,\Lambda^*)~.
			\end{aligned}
		\end{equation}
		
		\paragraph{Finite higher gauge transformations.}
		Higher gauge transformations in the base directions are parametrised by $h\in\scC^\infty(M,\sfH)$ and transform gauge transformations of $(A,A^*,B,B^*)$ parametrised by $(g,\Gamma^*,\Lambda,\Lambda^*)$ to new gauge transformations parametrised by
		\begin{equation}
			\begin{aligned}
				\tilde g\ &\coloneqq\ \sft(h)g~,
				\\
				\tilde\Gamma^*\ &\coloneqq\ \Gamma^*~,
				\\
				\tilde\Lambda\ &\coloneqq\ \Lambda+g^{-1}\acton(h^{-1}\rmd h+h^{-1}Ah)~,
				\\
				\tilde\Lambda^*\ &\coloneqq\ \Lambda^*~.
			\end{aligned}
		\end{equation}
		This transformation again composes associatively and has inverses.
		
		Mathematically, we have defined a higher action groupoid, describing connections, finite gauge transformations and finite higher gauge transformations, together with their compositions, inverses and actions.
		
		\paragraph{Half-adjusted differential cohomology.}
		Having identified the higher gauge action groupoid, we can use it to construct the complete half-adjusted differential cocycles and coboundaries. Concretely, consider a base manifold $M$ together with a surjective submersion $Y\rightarrow M$, e.g.~an atlas consisting of patches $Y=\bigsqcup_iU_i$ of open contractible subsets $U_i\subseteq M$. Let $Y^{[p]}$ denote the $p$-fold fibre product
		\begin{equation}
			Y^{[p]}\ \coloneqq\ \underbrace{Y\times_M Y\times_M\cdots\times_M Y}_{p{\rm-times}}~,
		\end{equation}
		i.e.~double, triple, etc.~overlaps of patches for $Y=\bigsqcup_iU_i$.
		
		We now use gauge transformations to glue gauge potentials on overlaps $Y^{[2]}$, and consistency of these yields cocycle relations for these gauge parameters. Gauge transformations on $Y$ that respect these consistency relations are called coboundaries; higher cocycles and coboundaries are derived analogously.
		
		The required computations are mostly straightforward, and the cocycle and coboundary relations in the base directions are already known from~\cite{Rist:2022hci}. We therefore merely list again the results.
		
		Half-adjusted differential cocycles with values in the half-adjusted cotangent crossed module of Lie groups $T^*[1]_\kappa(\sfH\rightarrow\sfG)$ consist of\footnote{For clarity, we add $p$ indices to the cocycle components to indicate that they are functions with domain $Y^{[p]}$ and to distinguish them from the coboundary components denoted by the same letters later.}
		\begin{subequations}\label{eq:adjustedCocycleConditions}
			\begin{equation}
				\begin{aligned}
					\hat h_{ijk}\ =\ h_{ijk}\ &\in\ \scC^\infty(Y^{[3]},\sfH)
					\\
					\hat g_{ij}\ =\ (g_{ij},\Gamma^*_{ij})\ &\in\ \scC^\infty(Y^{[2]},\sfG)\times\Omega^0(M)\otimes\frh^*
					\\
					\hat\Lambda_{ij}\ =\ (\Lambda_{ij},\Lambda^*_{ij})\ &\in\ \Omega^1(Y^{[2]})\otimes\frh\oplus\Omega^1(Y^{[2]})\otimes\frg^*~,
					\\
					\hat A_i\ =\ (A_i,A^*_i)\ &\in\ \Omega^1(Y^{[1]})\otimes\frg\oplus\Omega^1(Y^{[1]})\otimes\frg^*~,
					\\
					\hat B_i\ =\ (B_i,B^*_i)\ &\in\ \Omega^2(Y^{[1]})\otimes\frh\oplus\Omega^2(Y^{[1]})\otimes\frg^*~.
				\end{aligned}
			\end{equation}
			On appropriate double overlaps, these connection forms glue together as follows
			\begin{equation}
				\begin{aligned}
					A_j\ &=\ g_{ij}^{-1}A_ig_{ij}+g_{ij}^{-1}\rmd g_{ij}-\sft(\Lambda_{ij})~,
					\\
					A^*_j\ &=\ \nu^*_2(g_{ij}^{-1},A^*_i+\rmd\Gamma^*_{ij}+\nu^*_2(A_i,\Gamma^*_{ij}))-\sft^*(\Lambda^*_{ij})~,
					\\
					B_j\ &=\ g_{ij}^{-1}\acton B_i+\rmd\Lambda_{ij}+A_j^0\acton\Lambda_{ij}+\tfrac12[\Lambda_{ij},\Lambda_{ij}]-\kappa\big(g_{ij}^{-1},F_i\big)\,,
					\\
					B^*_j\ &=\ \tilde\nu^*_2\big(g_{ij}^{-1},B^*_i+\kappa^*_+(F_i,\Gamma^*_{ij})\big)+\rmd \Lambda^*_{ij}+\tilde\nu^*_2(A_j,\Lambda^*_{ij})~,
				\end{aligned}
			\end{equation}
			and consistency of these transformations on appropriate triple overlaps implies that
			\begin{equation}
				\begin{aligned}
					g_{ik}\ &=\ \sft(h_{ijk})g_{ij}g_{jk}~,
					\\
					\Lambda_{ik}\ &=\ \Lambda_{jk}+g_{jk}^{-1}\acton\Lambda_{ij}-g_{ik}^{-1}\acton(h_{ijk}\nabla_ih_{ijk}^{-1})~,
					\\
					\Gamma^*_{ik}\ &=\ \Gamma^*_{ij}+\nu^*_2(g_{ij},\Gamma^*_{jk})~,
					\\
					\Lambda^*_{ik}\ &=\ \Lambda^*_{jk}+\tilde\nu^*_2\big(g_{jk}^{-1},\Lambda^*_{ij}\big)~,
				\end{aligned}
			\end{equation}
			where we allowed for an additional 2-transformation in the gluing of the $g_{ij}$, and $\nabla_i\coloneqq\rmd+A_i\acton$. This 2-transformation then needs to satisfy the following consistency relation on appropriate quadruple overlaps
			\begin{equation}
				h_{ikl}h_{ijk}\ =\ h_{ijl}(g_{ij}\acton h_{jkl})~.
			\end{equation}
		\end{subequations}
		
		Coboundaries between two such cocycles $(\hat h_{ijk},\hat g_{ij},\hat\Lambda_{ij},\hat A_i,\hat B_i)$ and $(\tilde{\hat h}_{ijk},\tilde{\hat g}_{ij},\tilde{\hat\Lambda}_{ij},\tilde{\hat A}_i,\tilde{\hat B}_i)$ are parametrised by
		\begin{subequations}\label{eq:adjustedCoboundaryRelations}
			\begin{equation}
				\begin{aligned}
					\hat h_{ij}\ =\ h_{ij}\ &\in\ \scC^\infty(Y^{[2]},\sfH)~,
					\\
					\hat g_i\ =\ (g_i,\Gamma^*_i)\ &\in\ \scC^\infty(Y,\sfG)\times\Omega^0(Y)\otimes\frh^*~,
					\\
					\hat\Lambda_i\ =\ (\Lambda_i,\Lambda^*_i)\ &\in\ \Omega^1(Y)\otimes\frh\oplus\Omega^1(Y)\otimes\frg^*~,
				\end{aligned}
			\end{equation}
			and link the cocycles according to
			\begin{equation}\label{eq:adjustedCoboundaryConditionsB}
				\begin{aligned}
					\tilde h_{ijk}\ &=\ g_i^{-1}\acton(h_{ik}h_{ijk}(g_{ij}\acton h_{jk}^{-1})h_{ij}^{-1})~,
					\\
					\tilde g_{ij}\ &=\ g_i^{-1}\sft(h_{ij})g_{ij}g_j~,
					\\
					\tilde\Gamma^*_{ij}\ &=\ -\Gamma^*_i+\Gamma^*_{ij}+\nu^*_2(g_{ij},\Gamma^*_j)~,
					\\
					\tilde\Lambda_{ij}\ &=\ g_j^{-1}\acton\Lambda_{ij}+\Lambda_j-\tilde g_{ij}^{-1}\acton\Lambda_i+(g_j^{-1}g_{ij}^{-1})\acton(h_{ij}^{-1}\nabla_ih_{ij})~,
					\\
					\tilde\Lambda^*_{ij}\ &=\ \tilde\nu^*_2\big(g_j^{-1},\Lambda^*_{ij}\big)+\Lambda^*_j-\tilde\nu^*_2\big(\tilde g_{ij}^{-1},\Lambda^*_i\big)\,,
					\\
					\tilde A_i\ &=\ g_i^{-1}A_ig_i+g_i^{-1}\rmd g_i-\sft(\Lambda_i)~,
					\\
					\tilde A^*_i\ &=\ \nu^*_2\big(g_i^{-1},A^*_i+\rmd\Gamma^*_i+\nu^*_2(A_i,\Gamma^*_i)\big)-\sft^*(\Lambda^*_i)~,
					\\
					\tilde B_i\ &=\ g_i^{-1}\acton B_i+\rmd\Lambda_i+\tilde A_i\acton\Lambda_i+\tfrac12[\Lambda_i,\Lambda_i]-\kappa\big((g_i)^{-1},F_i\big)\,,
					\\
					B^*_j\ &=\ \tilde\nu^*_2\big(g_i^{-1},B^*_i+\kappa^*_+(F_i,\Gamma^*_i)\big)+\rmd \Lambda^*_i+\tilde\nu^*_2(A_j,\Lambda^*_i)~.
				\end{aligned}
			\end{equation}
		\end{subequations}
		The higher coboundary relations can be derived analogously.
		
		Altogether, the above cocycles and coboundaries describe general principal 2-bundles with half-adjusted connections as well as their bundle isomorphisms.
		
		\paragraph{Example: strict string Lie 2-group.} The most relevant example in the finite case is, as expected, the integrated string Lie 2-algebra. For simplicity however\footnote{The technicalities when using the weak Lie 2-group are substantial.}, we consider the strict but infinite-dimensional model~\cite{Baez:2005sn}. Consider a Lie group $\sfG$ with Lie algebra $\frg$, together with its based (parametrised) path and loop spaces $P_0\sfG$ and $L_0\sfG$. There is a central extension $\widehat{L_0\sfG}$ of $L_0\sfG$, which forms a principal circle bundle over $L_0\sfG$. This data can be combined into the crossed module of Lie groups
		\begin{equation}
			\caG\ \coloneqq\ \big( \widehat{L_0\sfG}~~\xrightarrow{~\sft~} P_0\sfG\big)~.
		\end{equation}
		
		As shown in~\cite{Rist:2022hci}, this crossed module can be adjusted with adjustment datum
		\begin{equation}\label{eq:adjustmentStringGroup}
			\begin{aligned}
				\hat\kappa\,:\,P_0\sfG\times P_0\frg\ &\rightarrow\ L_0\frg\oplus\fru(1)~,
				\\
				(g,X)\ &\mapsto\ \left((\id-\wp\circ \flat)(gXg^{-1}-X)
				\,,\,
				\frac{\rmi}{2\pi}\int_0^1\rmd r\,\innerLarge{g^{-1}\parder[g]{r}}{X}\right)
			\end{aligned}
		\end{equation}
		for all $g\in P\sfG$ and for all $V\in P_0\frg$, where $\flat:P_0\frg\rightarrow \frg$ is the end-point evaluation of the based path, and $\wp:[0,1]\rightarrow \IR$ is a function with $\wp(0)=0$ and $\wp(1)=1$.
		
		One can now construct both four-dimensional higher Chern--Simons theories with the corresponding half-adjustment as well as principal 2-bundles with half-adjusted connections. The latter are interesting, as they contain a generalisation of the string structures, which are in turn higher simultaneous analogues of instanton and monopole bundles, see~\cite{Rist:2022hci} for more details.
		
		Another potentially interesting adjusted crossed module to consider as a starting point for either half-adjusted Chern--Simons theory or principal 2-bundles with half-adjusted connections is the adjusted T-duality configuration 2-group $\sfT\sfD_n$ introduced in~\cite{Kim:2022opr}.
		
		\subsection{Higher-dimensional generalisation}
		
		Let us close with a discussion of the generalisation of half-adjusted higher Chern--Simons theory to higher dimensions. Our considerations are necessarily restricted to local, infinitesimal descriptions, because explicit forms of higher adjusted principal fibre bundles are complicated, and only available for principal 2- and 3-bundles~\cite{Rist:2022hci,Gagliardo:2025oio}.
		
		\paragraph{Half-adjusted $L_\infty$-algebra and action functional.} The definition of a general half-adjusted $L_\infty$-algebra is straightforward. Consider an adjusted $(d-2)$-term $L_\infty$-algebra $\frL$ and construct the cotangent completion $\hat\frL=T^*[d-3]\frL$ to a cyclic $L_\infty$-algebra, as given in~\eqref{eq:shifted_cotangent_construction}.
		
		In order to establish the results in the following, it will be useful to consider the Weil algebra of $\hat\frL$ in some detail. As before, let us denote the generators of $\sfW(\frL)$ by $\tte^A$ and $\hat\tte^A$. The generators of $\sfW(\hat\frL)$ are then those of $\sfW(\frL)$ together with a dual set $\ttE_A$ and $\hat\ttE_A$ with
		\begin{equation}
			|\ttE_A|\ =\ n-(|\tte^A|-1)
			\eand
			|\hat\ttE_A|\ =\ |\ttE_A|+1~.
		\end{equation}
		The adjustment of $\frL$ is given by a Weil algebra automorphism\footnote{See the proof of \cref{thm:only_abelian_cyclic_adjustments} for our notation.}
		\begin{subequations}
			\begin{equation}
				\tilde{\hat\tte}^A\ =\ \hat\tte^A+\kappa_{IJ}{}^A\tte^I{\hat\tte}^J~,
			\end{equation}
			and we will allow for a further automorphism
			\begin{equation}\label{eq:cotangent_automorphism}
				\tilde{\hat\ttE}_A\ =\ \hat\ttE_A+\lambda_{IA\underline{J}}{}^B\tte^I\tilde{\hat\tte}^{\underline{J}}\,{\hat\ttE}_B+\rho_{\underline{I}AJ}{}^B\tte^{\underline{I}}\,\tilde{\hat\tte}^{J}\ttE_B
			\end{equation}
		\end{subequations}
		with $\lambda_{IA\underline{J}}{}^B$ and $\rho_{\underline{I}AJ}{}^B$ some structure constants that we will specify later. By~\cref{thm:only_abelian_cyclic_adjustments}, it is clear that this does not yield an adjustment of $\hat \frL$ in general.
		
		\paragraph{Action functional.}
		Let $M$ be a $d$-dimensional compact oriented manifold without boundary. The field content consists of $\hat \frL$-valued forms given by a morphism $\caA:\sfW(\hat\frL)\rightarrow\Omega^\bullet(M)$. Concretely
		\begin{equation}
			\caA^A\ \coloneqq\ \caA(\tte^A)
			\eand
			\caA^*_A\ \coloneqq\ \caA(\ttE_A)
		\end{equation}
		with corresponding curvature forms
		\begin{equation}
			\caF^A\ \coloneqq\ \caA(\tilde{\hat\tte}^A)
			\eand
			\caF^*_A\ \coloneqq\ \caA(\tilde{\hat\ttE}_A)~.
		\end{equation}
		These combine into the action functional
		\begin{equation}\label{eq:action}
			S\ \coloneqq\ \int_M\caA^*_A\caF^A~.
		\end{equation}
		Upon varying this action with respect to the $\caA^*_A$, we recover half the desired equations of motion, $\caF^A=0$.
		
		\paragraph{Curvature forms.}
		To define the remaining curvatures $\caF^*_A$, which amounts to specifying the deformation parameters $\lambda$ and $\rho$ in~\eqref{eq:cotangent_automorphism}, we proceed as before in \cref{ssec:four_dimensional_case_infinitesimal_considerations}. That is, we consider the extension of the above fields and the Lagrangian $d$-form to an $(d+1)$-dimensional manifold $N$ which has $M$ as its boundary and extend all the fields from $M$ to $N$. We then demand that
		\begin{equation}\label{eq:curV_condition}
			\rmd(\caA^*_A\caF^A)\ \overset{!}{=}\ \caF^*_A\caF^A~,
		\end{equation}
		which is the evident generalisation of~\eqref{eq:cond_higher_pontryagin}. Note that
		\begin{equation}
			\rmd (\caA^*_A\caF^A)\ =\ (\rmd\caA^*_A)\caF^A+(-1)^{|A|}\caA^*_A\rmd\caF^A~,
		\end{equation}
		and the form of the Weil algebra~\eqref{eq:diff_form} ensures that
		\begin{equation}
			\rmd\caF^A\ =\ (-1)^{|B|}\caF^Bp_B{}^A(\caA,\caF)
		\end{equation}
		for $p_B{}^A(\caA,\caF)$ some monomial in the connections and curvature forms $\caA$ and $\caF$ so that~\eqref{eq:curV_condition} always has a solution.
		
		Whilst generically~\eqref{eq:curV_condition} does not uniquely determine $\caF^*_A$, it fixes it up to a field redefinition involving lower-dimensional curvatures. Therefore the choice of $\caF^*_A$ satisfying~\eqref{eq:curV_condition} translates into a choice of the deformation parameters $\lambda$ and $\rho$ in~\eqref{eq:cotangent_automorphism}. Different choices lead to equivalent descriptions of gauge configurations, as familiar from the tensor hierarchies, cf.~e.g.~\cite{Borsten:2021ljb}.
		
		We then have the following, desired results.
		\begin{proposition}
			Condition~\eqref{eq:curV_condition} implies that the action~\eqref{eq:action} is invariant under infinitesimal gauge transformations.
		\end{proposition}
		\begin{proof}
			The proof is fully analogous to the proof of \ref{prop:4d_gauge_invariance}.
		\end{proof}
		
		\begin{proposition}
			The equations of motion of the action~\eqref{eq:action} imply total flatness, i.e.~they amount to all curvatures vanishing. 
		\end{proposition}
		\begin{proof}
			Total flatness is equivalent to 
			\begin{equation}\label{eq:higher_eom}
				\delta\int_M\caA_A^*\caF^A\ =\ \int_M\Big\{(\delta\caA_A^*)\caF^A+(\caF^*_A+R^*_A)(\delta\caA^A)\Big\}~,
			\end{equation}
			where $\caF^*_A$ is any solution to~\eqref{eq:curV_condition} and $R^*_A$ is some polynomial in the fields containing at least one curvature form of lower degree than $\caF^*_A$.
			
			Firstly, we define
			\begin{equation}
				\caI^A\ \coloneqq\ \caF^A-\rmd\caA^A
				\eand
				\caI^*_A\ \coloneqq\ \caF^*_A-\rmd\caA^*_A~.
			\end{equation}
			Then, the condition~\eqref{eq:curV_condition} is equivalent to
			\begin{equation}
				(-1)^{|A|}\caA^*_A\rmd\caF^A\ =\ (-1)^{|A|}\caA^*_A\rmd\caI^A\ =\ \caI_A^*\caF^A
			\end{equation}
			on $N$. We now specialise to $N=M\times[0,1]$ with the interval $[0,1]$ coordinatised by $t$ and a variation as an infinitesimal homotopy. That is, we consider $\caA^A$ and $\caA^*_A$ on $N$ with
			\begin{equation}
				\delta(\caA^A|_M,\caA^*_A|_M)\ \coloneqq\ \left.\parder{t}\right|_{t=0}(\caA^A,\caA^*_A)
				\eand
				\parder{t}\intprod(\caA^A,\caA^*_A)\ =\ (0,0)~.
			\end{equation}
			It then follows that the Lie derivative along $\parder{t}$ at $t=0$ describes variations,
			\begin{equation}
				\begin{aligned}
					\left.\caL_{\parder{t}}\right|_{t=0}(\caA^A,\caA^*_A)\ &=\ \left.\parder{t}\right|_{t=0}\intprod\rmd(\caA^A,\caA^*_A)\ =\ \left.\parder{t}\right|_{t=0}(\caA^A,\caA^*_A)\ =\ \delta(\caA^A|_M,\caA^*_A|_M)~,
					\\
					\left.\caL_{\parder{t}}\right|_{t=0}\rmd(\caA^A,\caA^*_A)\ &=\ \rmd\left.\parder{t}\right|_{t=0}\intprod\rmd(\caA^A,\caA^*_A)\ =\ \rmd\left.\parder{t}\right|_{t=0}(\caA^A,\caA^*_A)\ =\ \rmd\delta(\caA^A|_M,\caA^*_A|_M)~,
				\end{aligned}
			\end{equation}
			and hence,
			\begin{equation}
				\delta(\caF^A,\caF^*_A)\ =\ \left.\caL_{\parder{t}}\right|_{t=0}(\caF^A,\caF^*_A)~.
			\end{equation}
			Consequently,
			\begin{equation}\label{eq:variation}
				\begin{aligned}
					\delta(\caA_A^*\caF^A)\ &=\ \left.\caL_{\parder{t}}\right|_{t=0}(\caA_A^*\caF^A)
					\\
					&=\ \left(\left.\parder{t}\intprod\rmd\right|_{t=0}+\left.\rmd \parder{t}\intprod\right|_{t=0}\right)(\caA_A^*\caF^A)
					\\
					&=\parder{t}\intprod (\caF_A^* \caF^A)+\rmd \parder{t}\intprod (\caA_A^*\caF^A)
					\\
					&=(\delta\caA_A^*)\caF^A\pm\caF_A^*\delta\caA^A+(\parder{t}\intprod \caI_A^*)\caF^A\pm\caF_A^*(\parder{t}\intprod \caI^A)\pm\rmd (\caA_A^* \parder{t}\intprod \caF^A)~,
				\end{aligned}
			\end{equation}
			where we used~\eqref{eq:curV_condition} from the second to the third line. The last term in the last line is a total derivative, and the other two terms are expressions proportional to curvature forms of lower degree. Moreover, these two terms are of the form $R_A^*\delta \caA^A$, since evidently
			\begin{equation}
				\delta(\caA_A^*\caF^A)\ =\ (\delta\caA_A^*)\caF^A+(...)\delta \caA^A~,
			\end{equation}
			and the first term on the right hand side is already present in the final line of~\eqref{eq:variation}. We hence conclude that~\eqref{eq:higher_eom} is indeed true.
		\end{proof}
		
		\paragraph{Closure of gauge transformations.}
		In order to study the closure of gauge transformations, recall from~\cref{prop:adjustment_L_infty} that we have to consider $\sfd_\sfW \tilde{\hat\tte}_A$. From~\eqref{eq:curV_condition}, we have that
		\begin{equation}
			(\sfd_\sfW\tilde{\hat\ttE}_A)\tilde{\hat\tte}^A\pm\tilde{\hat\ttE}_A(\sfd_\sfW\tilde{\hat\tte}^A)\ =\ 0~.
		\end{equation}
		This defines $\sfd_\sfW\tilde{\hat\ttE}_A$ up to terms of the form
		\begin{equation}
			R^*_A\ =\ q_{[AB]\cdots}\tilde{\hat\tte}^B\cdots~,
		\end{equation}
		which are graded-antisymmetric in $AB$. Translating this back into potentially non-vanishing contributions to $Q_{\rm BRST}\,\phi(\hat\tte^A)$ with $\phi$ the adjustment automorphism as explained in \cref{app:proof}, we see that the unambiguous terms do not contribute to this expression, but the terms proportional to $R^*_A$ may contain an arbitrary power of $\tte^A$, producing non-vanishing contributions to $Q_{\rm BRST}\,\phi(\hat\tte^A)$.
		
		Therefore, it is not possible to make any definite statement about closure of even ordinary gauge transformations in the higher case. However, it is clear that choosing adjusted $L_\infty$-algebras $\frL$ with structure constants of short length (i.e.~strict $L_\infty$-algebras with adjustment automorphisms that are at most binary) increases the chances of gauge transformations to close. It may therefore not be a coincidence that the largest class of examples of adjusted $L_\infty$-algebras known from the physical literature namely the ones appearing in the tensor hierarchies of gauged supergravity~\cite{Borsten:2021ljb} are exactly of this type.
		
		\section{Further constructions}
		
		In this final section, we shall comment on further ways of addressing the lack of adjusted, cyclic $L_\infty$-algebras identified in \cref{ssec:no_adjusted_cyclic_L-infty}, both in higher Chern--Simons theory as well as in general higher gauge theories.
		
		\subsection{Trivial symmetries}\label{ssec:trivial_symmetries}
		
		\paragraph{Generalities.}
		One reason for introducing adjustments is the failure of the off-shell closure of arbitrary higher gauge symmetries. Another possibility is to simply introduce new higher gauge symmetries that are suitable for compensating for the lack of closure. These are called \uline{trivial symmetries}~\cite[Section 3.1.5]{Henneaux:1992} as they vanish identically on solutions to the equation of motion~\eqref{eq:4d_curvatures} and therefore do not lead to conserved quantities. In the following, we discuss these trivial symmetries in detail. The upshot is that introducing trivial symmetries can simply be seen as excluding higher gauge symmetries.
		
		\paragraph{(Higher) trivial symmetries.}
		In~\eqref{eq:gauge_trafo_not_closing}, we have seen that the commutator of two gauge transformations~\eqref{eq:4dGaugeTransformations} is no longer a gauge transformation because of the extra terms
		\begin{equation}\label{eq:4dTrivialSymmetry}
			\Delta_{\alpha,\lambda;\alpha',\lambda'}A\ \coloneqq\ 0
			\eand
			\Delta_{\alpha,\lambda;\alpha',\lambda'}B\ \coloneqq\ \mu_3(F,\alpha,\alpha')~.
		\end{equation}
		
		We encounter a similar issue for higher gauge transformations. Indeed, upon performing gauge transformations~\eqref{eq:4dGaugeTransformations} with gauge parameters $(\alpha,\lambda)$ and with the ones $(\alpha+\mu_1(\sigma),\rmd\sigma+\mu_2(A,\sigma))$ obtained from higher gauge transformations~\eqref{eq:4dHigherGaugeTransformation}, respectively, we obtain different gauge-transformed gauge potentials with the difference being
		\begin{equation}\label{eq:4dDiscrepancyGaugeTransformedGaugeParameters}
			\Delta_\sigma A\ \coloneqq\ 0
			\eand
			\Delta_\sigma B\ \coloneqq\ \mu_2(F,\sigma)~,
		\end{equation}
		see~\eqref{eq:higher_gauge_trafo_not_closing}.
		
		It is not difficult to see that the higher Chern--Simons action~\eqref{eq:4dCSAction} is invariant under the transformation~\eqref{eq:4dTrivialSymmetry}. We can therefore introduce corresponding new transformations~\eqref{eq:4dTrivialSymmetry} and~\eqref{eq:higher_gauge_trafo_not_closing}. These will be trivial symmetries, as they vanish on the solutions to $F=0$. 
		
		More generally, the commutator of two gauge transformations~\eqref{eq:gaugeTransformationsHigherCSTheory} is again a gauge transformation up to the term\footnote{See~\cite[Appendix C]{Jurco:2018sby} for the explicit computation in the general case.}
		\begin{equation}\label{eq:commutatorGaugeTransformationsA}
			\Delta_{c_0;c'_0}a\ \coloneqq\ \sum_{i\geq0}\frac1{i!}\mu^{\Omega^\bullet(M,\frL)}_{i+3}(a,\ldots,a,f,c_0,c'_0)~,
		\end{equation}
		and it is easily seen that the higher Chern--Simons action~\eqref{eq:higherCSAction} is invariant under these transformations as this follows directly from the cyclicity of the inner product and the fact that $\mu_{i+2}(f,f,\ldots)=0$ for all $i\in\IN_0$. We can generalise~\eqref{eq:commutatorGaugeTransformationsA} to very general trivial symmetries of the higher Chern--Simons action given by
		\begin{equation}\label{eq:trivialSymmetries}
			\Delta_{V_1;\ldots;V_i}a\ \coloneqq\ \sum_{j\geq0}\frac1{j!}\mu^{\Omega^\bullet(M,\frL)}_{i+j+1}(a,\ldots,a,f,V_1,\ldots,V_i)
		\end{equation}
		for all homogeneous $V_1,\ldots,V_i\in\Omega^\bullet(M,\frL)$ with $\sum_{j=1}^i|V_j|=i-2$. Likewise, we have \uline{higher trivial symmetries} which are given by
		\begin{equation}\label{eq:higherTrivialSymmetries}
			\Delta_{V_1,\ldots,V_i}c_{-k}\ \coloneqq\ \sum_{j\geq0}\frac1{j!}\mu^{\Omega^\bullet(M,\frL)}_{i+j+1}(a,\ldots,a,f,V_1,\ldots,V_i)
		\end{equation}
		for all homogeneous $V_1,\ldots,V_i\in\Omega^\bullet(M,\frL)$ with $\sum_{j=1}^i|V_j|=i-k-3$.
		
		\paragraph{Example: Strict $2$-term $L_\infty$-algebras.}
		In order to understand the implication of introducing these additional trivial symmetries, let us consider the special case of a strict 2-term $L_\infty$-algebra. Evidently, in this case, the trivial symmetries~\eqref{eq:4dTrivialSymmetry} are vacuous, and the gauge transformations do close. The discrepancy~\eqref{eq:4dDiscrepancyGaugeTransformedGaugeParameters} between performing gauge transformations with gauge parameters and higher-gauge-transformed gauge parameters however remains. We compensate by introducing an additional gauge transformation as follows.
		
		Consider again the $d=4$ higher Chern--Simons theory but with now a strict $2$-term $L_\infty$-algebra. We enlarge $\Omega^\bullet(M,\frL)$ to
		\begin{equation}
			\begin{aligned}
				\tilde\Omega^\bullet_1(M,\frL)\ &\coloneqq\ \Omega^1(M)\otimes\frL_0\oplus\Omega^2(M)\otimes\frL_{-1}~,
				\\
				\tilde\Omega^\bullet_0(M,\frL)\ &\coloneqq\ \Omega^0(M)\otimes(\frL_{-1}\oplus\frL_0)\oplus\Omega^1(M)\otimes\frL_{-1}~,
				\\
				\tilde\Omega^\bullet_{-1}(M,\frL)\ &\coloneqq\ \Omega^0(M)\otimes\frL_{-1}
			\end{aligned}
		\end{equation}
		and with $\mu_i^{\tilde\Omega^\bullet(M,\frL)}=\mu_i^{\Omega^\bullet(M,\frL)}$ for $i=1,2$. Evidently, as the vector space of elements of degree one does not change, the higher Chern--Simons action does not change. Hence, the curvatures~\eqref{eq:4d_curvatures} are not altered either. What changes, however, are the gauge transformations and the higher gauge transformations.\footnote{Note that we are clearly losing the interpretation of gauge transformations as partially flat homotopies.} In particular, the space of gauge parameters has been enlarged, and we enlarge the gauge transformations~\eqref{eq:4dGaugeTransformations} by a trivial symmetry, that is,
		\begin{equation}
			\begin{aligned}
				\delta_{\alpha,\gamma,\lambda}A\ &=\ \rmd \alpha+\mu_2(A,\alpha)-\mu_1(\lambda)~,
				\\
				\delta_{\alpha,\gamma,\lambda}B\ &=\ \mu_2(B,\alpha)+\rmd\lambda+\mu_2(A,\lambda)+\mu_2(F,\gamma)~.
			\end{aligned}
		\end{equation}
		It is easy to see that this can be rewritten as
		\begin{subequations}
			\begin{equation}
				\begin{aligned}
					\delta_{\alpha,\gamma,\lambda}A\ &=\ \rmd \alpha'+\mu_2(A,\alpha')-\mu_1(\lambda')~,
					\\
					\delta_{\alpha,\gamma,\lambda}B\ &=\ \mu_2(B,\alpha')+\rmd\lambda+\mu_2(A,\lambda')
				\end{aligned}
			\end{equation}
			with
			\begin{equation}\label{eq:twistedCombinationGaugeParameters}
				\alpha'\ \coloneqq\ \alpha+\mu_1(\gamma)
				\eand
				\lambda'\ \coloneqq\ \lambda+\rmd\gamma+\mu_2(A,\gamma)~.
			\end{equation}
		\end{subequations}
		Hence, the commutator of two such enlarged gauge transformations is again an enlarged gauge transformation. Furthermore, the higher gauge transformations~\eqref{eq:4dHigherGaugeTransformation} change accordingly to
		\begin{equation}
			\begin{aligned}
				\delta_\sigma \alpha\ &\coloneqq\ \mu_1(\sigma)~,
				\\
				\delta_\sigma\gamma\ &\coloneqq\ -\sigma~,
				\\
				\delta_\sigma\lambda\ &\coloneqq\ \rmd\sigma+\mu_2(A,\sigma)~.
			\end{aligned}
		\end{equation}
		With these higher gauge transformation, it is easy to see that the combination of gauge parameters~\eqref{eq:twistedCombinationGaugeParameters} remains invariant. Consequently, it does not matter if one performs gauge transformations with gauge parameters $(\alpha,\gamma,\lambda)$ or $(\alpha+\mu_1(\sigma),\gamma-\sigma,\lambda+\rmd\sigma+\mu_2(A,\sigma))$ and the discrepancy~\eqref{eq:4dDiscrepancyGaugeTransformedGaugeParameters} has disappeared.
		
		Considering the underlying $L_\infty$-algebra $\tilde\Omega^\bullet(M,\frL)$, however, we see that the higher gauge parameter $\sigma$ and the additional gauge parameter $\gamma$ are linked by an identity map. This means that they essentially form a ``trivial pair'' in BRST parlance and therefore drop out of the cohomology or the minimal model. Hence, there is a quasi-isomorphism from $\tilde\Omega^\bullet(M,\frL)$ to another $L_\infty$-algebra with both the trivial symmetries and the higher gauge symmetries removed. Altogether, introducing trivial symmetries as additional gauge symmetries simply removes the higher gauge transformations, turning a higher gauge theory into an ordinary gauge theory with higher-form gauge potentials.
		
		\paragraph{General case.} It is now rather clear how to introduce trivial symmetries in a higher gauge theory: for each failure of closure of the higher gauge algebra, introduce an additional gauge freedom to compensate. In doing so, however, one would want to ensure that quasi-isomorphic gauge $L_\infty$-algebra then lead to semi-classically equivalent higher Chern--Simons theories\footnote{i.e~theories with quasi-isomorphic field theory $L_\infty$-algebra describing the BV theory}. Currently, it is not clear to us if this can always be achieved. Another concern is that gauge transformations can no longer be seen as partially flat homotopies, a key concept within gauge theory.
		
		Beyond this, we note that whilst adding trivial symmetries ensures that (higher) gauge symmetries close, other problems remain. For example, the self-duality equation for higher $(d/2)$-form curvatures $F_{d/2}=\star F_{d/2}$ in $d$ dimensions is generically not gauge covariant. This equation naturally arises in many contexts within string theory and supergravity, e.g.~in six-dimensional superconformal field theories, which involve the tensor multiplet.
		
		Altogether, we are left with the impression that trivial symmetries are less natural than half-adjustments. Adding to this impression is the observation that the higher gauge theories known from physics all make use of adjustments.
		
		\subsection{General higher gauge theories}\label{ssec:general_higher_gauge}
		
		It may seem strange that the lack of cyclic adjusted $L_\infty$-algebras has not shown up in other context in physics. After all, higher gauge theories are ubiquitous, e.g.~in the context of the tensor hierarchies of gauged supergravity or higher-dimensional superconformal field theories. Let us briefly comment on how these theories avoid using cyclic adjusted higher gauge Lie algebras.
		
		\paragraph{Theories not respecting quasi-isomorphisms.}
		There are a number of theories that are constructed in a way that a quasi-isomorphism of the gauge $L_\infty$-algebra would not automatically induce a quasi-isomorphism of the $L_\infty$-algebra describing the perturbative field theory. Hence, quasi-isomorphic gauge $L_\infty$-algebras do not lead to semi-classically equivalent field theories, and our argument in \cref{ssec:no_adjusted_cyclic_L-infty} cannot be generally applied.
		
		An example of such a field theory is the higher Stueckelberg model discussed in~\cite{Borsten:2024gox}, where the gauge $L_\infty$-algebra is quasi-isomorphic to a trivial $L_\infty$-algebra. Here, an adjustment was not necessary, since the theory is Abelian.
		
		Another example of such a field theory is the six-dimensional (1,0)-theory discussed in~\cite{Saemann:2017zpd,Saemann:2019dsl,Rist:2020uaa} (and based on earlier work~\cite{Samtleben:2011fj,Samtleben:2012fb}), where the higher gauge algebra is quasi-isomorphic to an ordinary Lie algebra. Here, an adjustment exists, but it does not have to respect cyclicity. In fact, the Lagrangian does not make use of a cyclic higher product; in particular, the pairing appearing in the Lagrangian is not graded symmetric. This brings us to the next class of examples.
		
		\paragraph{Theories without cyclic structure.}
		We mentioned before the adjusted string Lie 2-algebra naturally appears in the context of heterotic supergravity~\cite{Bergshoeff:1981um,Chapline:1982ww}, see also~\cite{Bergshoeff:1989de}. This action contains curvature terms of the form $\int \langle H,\star H\rangle$ and $\int \langle F,\star F\rangle$ with $F=\rmd A+\tfrac12[A,A]$ and $H=\rmd B+{\rm cs}(A)$ the components of the field strength introduced in~\eqref{eq:string_field_strength}. The pairing of two elements of degree $-1$ in the string Lie 2-algebra clearly indicates that the quadratic form $\langle-,-\rangle$ is not of fixed degree and hence not a cyclic structure in our sense.
		
		The same is true for gauged supergravities, cf.~e.g.~\cite{Samtleben:2008pe}, for which the Lagrangian involves terms of the form $\int \langle F_k,\star F_k\rangle$ with $F_k$ a $k$-form on space-time, again requiring pairings of indefinite $L_\infty$-algebra degree. As mentioned above, the underlying $L_\infty$-algebras arise from a shifted truncation of a differential graded Lie algebra, and for this class of $L_\infty$-algebras there is a homogeneous construction of adjustments~\cite{Borsten:2021ljb}. Again, we recognise a loophole to our argument in \cref{ssec:no_adjusted_cyclic_L-infty}.
		
		Another, closely related theory is the BF-type theory considered in~\cite[Section 6.6.1]{Sati:2008eg}. Again, no cyclicity was used in the construction of the Lagrangian. More generally, higher Chern--Simons theories in the sense of~\cite{Sati:2008eg,Fiorenza:2011jr,Fiorenza:2011jr} would be integrals over Chern--Simons terms, i.e.~Lagrangians that, once lifted to a higher-dimensional space, differentiate to a polynomial in curvature forms. An example of such a theory is the seven-dimensional Chern--Simons theory discussed in~\cite{Fiorenza:2012tb}. Note that this theory's equation of motion does not amount to full flatness of the involved connections.
		
		\paragraph{Theories ignoring higher gauge transformations.}
		Some literature simply ignores higher gauge transformations and their closure, see e.g.~the BF-type theory discussed in~\cite[Section 3.9]{Girelli:2003ev}. As we saw above, we can model such theories by introducing additional trivial symmetries. Also, such theories are essentially 	ordinary gauge theories of BF-type.
		
		\paragraph{Conclusion.}
		We summarise that the higher gauge theories arising in the context of string theory usually do not make use of cyclic $L_\infty$-algebras as their higher gauge algebras. This, however, is a necessity for higher Chern--Simons theories, at least for the theories arising in a straightforward manner from homotopy Maurer--Cartan theory. We therefore believe that half-adjusted higher Chern--Simons theory provides the most promising form of higher Chern--Simons theory for further study.
		
		\appendix
		\addappheadtotoc
		\appendixpage
		
		\appendices
		
		\section{Derivation of \texorpdfstring{\cref{prop:adjustment_L_infty}}{Proposition 3.1}}\label{app:proof}
		
		For the reader's convenience, let us briefly summarise the derivation of \cref{prop:adjustment_L_infty} also found in~\cite{Gagliardo:2025oio}.
		
		We start from the construction of the adjusted BRST complex~\cite{Saemann:2019dsl,Fischer:2024vak,Gagliardo:2025oio}. We extend the morphism~\eqref{eq:morphism_Weil} to an inner homomorphism in the category of N$Q$-manifolds\footnote{i.e.~differential graded manifolds concentrated in non-positive degrees}. Essentially, this means that we allow for maps
		\begin{equation}\label{eq:augmented_morphism}
			\caA\,:\,\sfW(\frL)\ \rightarrow\ \Omega^\bullet(M)
		\end{equation}
		of non-negative degree, called the \uline{ghost number}, and there is a differential on these maps of ghost degree one, the \uline{BRST differential}
		\begin{equation}
			Q_{\rm BRST}\caA\ \coloneqq\ \rmd\circ\caA-\caA\circ\sfd_\sfW~.
		\end{equation}
		The images of the generators $\tte^A$ and $\phi(\hat\tte^A)$ in the adjusted Weil algebra are of the form
		\begin{equation}
			\begin{aligned}
				\caA(\tte^A)\ &=\ a^A+c_0^A+c_{-1}^a+\cdots~,
				\\
				\caA(\phi(\hat\tte^A))\ &=\ f^A+d_0^A+d_{-1}^A+\cdots~,
			\end{aligned}
		\end{equation}
		where $a$ and $f$ are the sums of the connection and curvature forms, respectively, and $c_i$ and $d_i$ are elements of ghost number $1-i$. Whilst we expect the ghosts $c_i$, keeping the ghosts $d_i$ would lead to too many gauge symmetries: the configuration space is obtained by quotienting connections by the gauge symmetries and would thus be too small. Hence, we need to quotient the image of $\caA$ by the ideal generated by the $d_i$ and its derivatives. To be consistent with $Q_{\rm BRST}^2=0$, however, we need to require that the BRST differential on these ghosts vanishes. Explicitly, this means that
		\begin{equation}
			\sfp(Q_{\rm BRST}\caA(\phi(\hat\tte^A)))=0~,
		\end{equation}
		or, equivalently
		\begin{equation}
			\sfp(\caA(\sfd_\sfW(\phi(\hat\tte^A))))=0~,
		\end{equation}
		where $\sfp$ denotes the projection onto components of ghost number larger than one that are free of $d_i$.
		
		Monomials in $\sfW(\frL)$ which are mapped under $\caA$ to field monomials of ghost number larger than one either contain at least one generator $\tte^A$ of degree $\geq 2$ or at least two generators $\tte^A$ of degree $1$. This means that the monomials that are in the kernel of $\sfp\circ\caA$ are of the form
		\begin{equation}
			f_{B_1\cdots B_i}{}^A\phi(\hat\tte^{B_1})\cdots\phi(\hat\tte^{B_i})
			\eor
			g_{B_0B_1\cdots B_i}{}^A\tte^{B_0}\phi(\hat\tte^{B_1})\cdots\phi(\hat\tte^{B_i})
		\end{equation}
		with $|\tte^{B_0}|=1$. Thus, the automorphism $\phi$ is an adjustment if and only if the differential on the adjusted Weil algebra is of the form~\eqref{eq:diff_form}.
		
		\section{Details on the half-adjusted gauge structure}
		\label{app:half-adjusted-gauge}
		
		\paragraph{Properties of deformed actions.}
		At the finite level, we note that the maps $\nu_2$ and $\tilde\nu_2$ introduced in~\eqref{eq:nu_maps} satisfy the following relations:
		\begin{equation}
			\begin{aligned}
				\sft(\nu_2(g,Y))\ &=\ \tilde\nu_2(g,\sft(Y))~,
				\\
				\nu_2(\sft(h),Y)\ &=\ Y~,
				\\
				\tilde\nu_2(\sft(h),X)\ &=\ X~,
				\\
				\nu_2(g_1,\nu_2(g_2,Y))\ &=\ \nu_2(g_1g_2,Y)~,
				\\
				\tilde\nu_2(g_1,\tilde\nu_2(g_2,X))\ &=\ \tilde\nu_2(g_1g_2,X)
			\end{aligned}
		\end{equation}
		for all $Y\in\frh$, $h\in\sfH$, $X\in\frg$, and $g,g_{1,2}\in\sfG$. In particular, we note that $\nu_2$ and $\tilde \nu_2$ behave like left-actions on $\frh$ and $\frg$, respectively.
		
		We have further the identity
		\begin{equation}
			\begin{aligned}
				\kappa^*(g_1g_2,Y^*)(X)\ &=\ Y^*(\kappa(g_2^{-1}g_1^{-1},X))
				\\
				&=\ Y^*\big(g_2^{-1}\acton\kappa(g_1^{-1},X)+\kappa(g_2^{-1},\tilde\nu_2(g_1^{-1},X))\big)
				\\
				&=\ \kappa^*(g_1,g_2\acton Y^*)(X)+\tilde\nu^*_2(g_1,\kappa^*(g_2,Y^*))(X)~,
			\end{aligned}
		\end{equation}
		and similarly
		\begin{equation}
			\begin{aligned}
				\nu^*_2(g_1g_2,Y^*)\ &=\ \nu^*_2(g_1,g_2\acton Y^*)-\sft^*(\tilde\nu^*_2(g_1,\kappa^*(g_2,Y^*)))~,
				\\
				\tilde\nu^*_2(g_1g_2,X^*)\ &=\ g_1\tilde\nu^*_2(g_2,X^*)g_1^{-1}-\kappa^*(g_1,\nu^*_2(g_2,\mu^*_1(X^*)))
			\end{aligned}
		\end{equation}
		for all $Y^*\in\frh^*$, $X\in\frg$, $X^*\in\frg^*$, and $g_{1,2}\in\sfG$. These imply the relations
		\begin{equation}
			\begin{aligned}
				\rmd \kappa^*(g,Y^*)\ &=\ \kappa^*(g,g^{-1}\rmd g\acton Y^*)+\nu^*_2(g,\kappa^*(g^{-1}\rmd g,Y^*))+\kappa^*(g,\rmd Y^*)~,
				\\
				\rmd \nu^*_2(g,Y^*)\ &=\ \nu^*_2(g,\nu^*_2(g^{-1}\rmd g,Y^*))+\nu^*_2(g,\rmd Y^*)~,
				\\
				\rmd \tilde\nu^*_2(g,X^*)\ &=\ \tilde\nu^*_2(g,\tilde\nu^*_2(g^{-1}\rmd g,X^*))+\tilde\nu^*_2(g,\rmd X^*)
			\end{aligned}
		\end{equation}
		for the action of the de~Rham differential for all $Y^*\in\frh^*$, $X^*\in\frg^*$, and $g\in\sfG$. In addition, we also have
		\begin{equation}
			\begin{aligned}
				\sft^*(\kappa^*_+(X,Y^*))\ =\ -\tfrac12\nu^*_2(X,Y^*)
			\end{aligned}
		\end{equation}
		for all $X\in\frg$ and $Y^*\in\frh^*$.
		
		A particularly useful identity to show that finite gauge transformations close is the following one:
		\begin{equation}
			\nu_2(g,\kappa_+(X_1,X_2))\ =\ \kappa_+(\tilde\nu_2(g,X_1),\tilde \nu_2(g,X_2))
		\end{equation} 
		for $X_{1,2}\in \frg$ and $g\in \sfG$. To demonstrate this identity, we replace on both sides $\nu_2$ and $\kappa_+$ with their definitions in terms of $\kappa$. We then apply the two adjustment identities
		\begin{equation}
			\kappa(\sft(Y),X_1)\ =\ -X_1\acton Y
		\end{equation} 
		and
		\begin{equation}
			\begin{aligned}
				\kappa(X_1,\sft(\kappa(g,X_2))&\ =\ X_1\acton \kappa(g,X_2)+\kappa(X_1,gX_2g^{-1})-g\acton \kappa(g^{-1}X_1g,X_2)
				\\
				&\hspace{1cm}-\kappa\big(g,[g^{-1}X_1g,X_2]-\sft(\kappa([g^{-1}X_1g,X_2]))\big)
			\end{aligned}
		\end{equation}
		with $X_{1,2}\in \frg$, $Y\in \frh$, $g\in \sfG$ to the rewritten right-hand side. Note that the second identity arises from the adjustment condition~\eqref{eq:alternativeAdjustmentCondition_b} by linearising the argument $g_2$ to $\unit+X_1$.
		
		It remains to use the crossed module identity
		\begin{equation}
			\sft(Y_1)\acton Y_2\ =\ [Y_1,Y_2]
		\end{equation} 
		for $Y_{1,2}\in \frh$ to show equality of both sides.
		
		\paragraph{Composition of infinitesimal gauge transformations.}
		Let us briefly present the details underlying the relation~\eqref{eq:closure_of_gauge}. Consider again the vector of gauge potentials $a\in\{A,A^*,B,B^*\}$ as well as gauge transformations $\delta_i$ parametrised by $(\alpha_i,\alpha^*_i,\lambda_i,\lambda^*_i)$. We compute 
		\begin{equation}
			\begin{aligned}
				a+\delta_2a\ =\
				\begin{pmatrix}
					A+\rmd\alpha_2+\mu_2(A,\alpha_2)-\mu_1(\lambda_2)
					\\
					A^*+\rmd\alpha^*_2+\nu^*_2(A,\alpha^*_2)-\nu^*_2(\alpha_2,A^*)-\mu^*_1(\lambda^*_2)
					\\
					B+\rmd\lambda_2+\mu_2(A,\lambda_2)-\mu_2(\alpha_2,B)+\tfrac12\mu_3(A,A,\alpha_2)+\kappa(\alpha_2,F)
					\\
					B^*+\rmd\lambda^*_2+\tilde\nu^*_2(A,\lambda^*_2)-\tilde\nu^*_2(\alpha_2,B^*)+\kappa^*_+(F,\alpha^*_2)
				\end{pmatrix}~,
			\end{aligned}
		\end{equation}
		and furthermore,
		\begin{equation}
			\begin{aligned}
				&\delta_1(a+\delta_2a)-(a+\delta_1a)
				\\
				&=\begin{pmatrix}
					\mu_2(\rmd\alpha_2+\mu_2(A,\alpha_2)-\mu_1(\lambda_2),\alpha_1)
					\\[5pt]
					\nu^*_2(\rmd\alpha_2+\mu_2(A,\alpha^*_2)-\mu_1(\lambda_2),\alpha^*_1)
					\\
					-\nu^*_2(\alpha_1,\rmd\alpha^*_2+\nu^*_2(A,\alpha^*_2)-\nu^*_2(\alpha_2^0,A^*)-\mu^*_1(\lambda^*_2))
					\\[5pt]
					\mu_2(\rmd\alpha_2+\mu_2(A,\alpha_2)-\mu_1(\lambda_2),\lambda_1)
					\\
					-\mu_2(\alpha_1,\rmd\lambda_2+\mu_2(A,\lambda_2)-\mu_2(\alpha_2,B)+\tfrac12\mu_3(A,A,\alpha_2)+\kappa(\alpha_2,F))
					\\
					+\mu_3(\rmd\alpha_2+\mu_2(A,\alpha_2)-\mu_1(\lambda_2),A,\alpha_1)-\kappa(\alpha_1,\tilde\nu_2(\alpha_2,F))
					\\[5pt]
					\tilde\nu^*_2(\rmd\alpha_2+\mu(A,\alpha_2)-\mu_1(\lambda_2),\lambda^*_1)
					\\
					-\tilde\nu^*_2(\alpha_1,\rmd\lambda^*_2+\tilde\nu^*_2(A,\lambda^*_2)-\tilde\nu^*_2(\alpha_2,B^*)+\kappa^*_+(F,\alpha^*_2))
					\\
					+\kappa^*_+(-\tilde\nu_2(\alpha_2,F),\alpha^*_1)
				\end{pmatrix}.
			\end{aligned}
		\end{equation}
		The commutator is then given by
		\begin{equation}
			\begin{aligned}
				[\delta_1,\delta_2]a\ &=\
				\begin{pmatrix}
					\rmd\alpha_3+\mu_2(A,\alpha_3)-\mu_1(\lambda_3)
					\\
					\rmd\alpha^*_3+\nu^*_2(A,\alpha^*_3)-\nu^*_2(\alpha_3,A^*)-\mu^*_1(\lambda^*_3)
					\\
					\rmd\lambda_3+\mu_2(A,\lambda_3)-\mu_2(\alpha_3,B)+\tfrac12\mu_3(A,A,\alpha_3)+\kappa(\alpha_3,F)
					\\
					\rmd\lambda^*_3+\tilde\nu^*_2(A,\lambda^*_3)-\tilde\nu^*_2(\alpha_3,B^*)+\kappa^*_+(F,\alpha^*_3)
				\end{pmatrix}
			\end{aligned}
		\end{equation}
		and so, we obtain~\eqref{eq:closure_of_gauge}.
		
	\end{body}
	
\end{document}